\documentclass[a4paper]{amsart}
\usepackage{amssymb}
\usepackage[many]{tcolorbox}
\usepackage{graphicx} 
\usepackage{mathrsfs}
\usepackage{braket}
\usepackage{hyperref}
\usepackage{mathtools}
\usepackage{silence}
\usepackage{microtype}
\usepackage{enumitem}
\usepackage{xcolor}
\usepackage{mdwtab}
\usepackage{complexity}
\usepackage[mathcal]{euscript}
\usepackage[all]{xy}
\usepackage[store-sets-label]{keytheorems}
\usepackage[nameinlink,capitalize]{cleveref}
\usepackage[noadjust]{cite}
\usepackage{comment}
\usepackage[all]{xy}
\usepackage[normalem]{ulem}
\usepackage{subfig}
\usepackage{afterpage}
\usepackage{xifthen}
\usepackage{zref-clever}
\zcsetup{nameinlink,cap}
\usepackage{rotating}

\newcommand{\pairs}[2]{(#1)_{#2}}

\zcRefTypeSetup{conj}{
Name-sg = Conjecture ,
name-sg = conjecture ,
Name-pl = Conjectures ,
name-pl = conjectures ,
}
\zcRefTypeSetup{fact}{
Name-sg = Fact ,
name-sg = fact ,
Name-pl = Facts ,
name-pl = facts ,
}
\zcRefTypeSetup{claim}{
Name-sg = Claim ,
name-sg = claim ,
Name-pl = Claims ,
name-pl = claims ,
}

\newcommand{\unsure}{\textcolor{orange}{\LARGE ?}}

\graphicspath{{Figures/}{../Figures/}}
\makeatletter
\newsavebox\myboxA
\newsavebox\myboxB
\newlength\mylenA

\newcommand*\xoverline[2][0.75]{%
    \sbox{\myboxA}{$\m@th#2$}%
    \setbox\myboxB\null
    \ht\myboxB=\ht\myboxA%
    \dp\myboxB=\dp\myboxA%
    \wd\myboxB=#1\wd\myboxA
    \sbox\myboxB{$\m@th\overline{\copy\myboxB}$}
    \setlength\mylenA{\the\wd\myboxA}
    \addtolength\mylenA{-\the\wd\myboxB}%
    \ifdim\wd\myboxB<\wd\myboxA%
       \rlap{\hskip 0.5\mylenA\usebox\myboxB}{\usebox\myboxA}%
    \else
        \hskip -0.5\mylenA\rlap{\usebox\myboxA}{\hskip 0.5\mylenA\usebox\myboxB}%
    \fi}
\makeatother

\hypersetup{colorlinks=true,linkcolor=blue,filecolor=blue,citecolor=blue}
\DeclareMathOperator{\mw}{mw}
\DeclareMathOperator{\width}{width}
\DeclareMathOperator{\reach}{MReach}

\DeclareMathOperator{\bomega}{\ensuremath{{\rm b}\omega}}
\DeclareMathOperator{\model}{{\bf Mod}}
\DeclareMathOperator{\Seq}{{\rm Seq}}
\DeclareMathOperator{\Str}{{\sf Str}}
\DeclareMathOperator{\Exp}{{\rm Exp}}
\DeclareMathOperator{\Skel}{{\rm Skel}}

\title{On merge-models}
\thanks{\ERCagreement}
\author[H. Buffière]{Hector Buffière}
\address{Université Paris Cité, CNRS, IRIF, Paris, France  \and Centre d'Analyse et de Mathématique Sociales CNRS UMR 8557, France}
\email{buffiere@irif.fr}
\author[Y. Lin]{Yuquan Lin}
\address{Southeast University, Nanjing, Jiangsu, China \and Centre d'Analyse et de Mathématique Sociales CNRS UMR 8557, France.}
\email{yqlin@seu.edu.cn}
\author[J. Ne\v set\v ril]{Jaroslav Ne\v set\v ril}
\address{Computer Science Institute of Charles University (IUUK), Praha, Czech Republic}
\email{nesetril@iuuk.mff.cuni.cz}
\author[P. Ossona de Mendez]{Patrice Ossona de Mendez}
\address{Centre d'Analyse et de Mathématique Sociales CNRS UMR 8557, France \and Computer Science Institute of Charles University (IUUK), Praha, Czech Republic}
\email{pom@ehess.fr}
\author[S. Siebertz]{Sebastian Siebertz}
\address{University of Bremen, Bremen, Germany}
\email{siebertz@uni-bremen.de}

\date{\today}
\newcommand{\ERCagreement}{
        {\footnotesize
        The second author is supported by the China Scholarship Council (CSC)	and  SEU Innovation Capability Enhancement Plan for Doctoral Students (CXJH\_SEU 24119).}\\[4pt]
		\noindent\begin{minipage}{.73\textwidth}
			\footnotesize
    This paper is part of a project that has received funding from the European Research Council (ERC) under the European Union's Horizon 2020 research and innovation program (grant agreement No 810115 -- {\sc Dynasnet}), and from the 
    German Research Foundation (DFG) with grant agreement No 446200270
		\end{minipage}\hfill
		\begin{minipage}{.25\textwidth}
\phantom{.}\hfill\includegraphics[height=13mm]{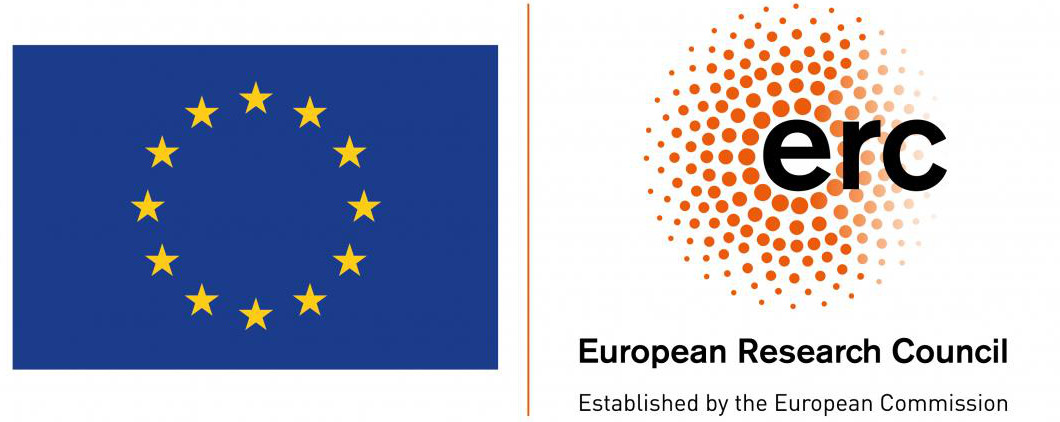}\hfill\phantom{.}
		\end{minipage}
}

\date{\today}

\newtheorem{fact}{Fact}[section]

\newtheorem{definition}{Definition}[section]

\newtheorem{claim}{$\rhd$ Claim}[section]

\newtheorem{ndefi}[definition]{Definition}
 
\newkeytheoremstyle{tcbext}{
tcolorbox-no-titlebar={breakable,before skip=5pt,after skip=5pt,	outer arc=0pt, blanker,
borderline west={1pt}{0pt}{white!20!gray},
	top=3pt,
	right=3pt,
    left=5pt,
	bottom=3pt},
}
\newkeytheoremstyle{tcbconj}{
tcolorbox-no-titlebar={breakable,before skip=5pt,after skip=5pt,	outer arc=0pt, blanker,
borderline west={1pt}{0pt}{purple!20!gray},
	top=3pt,
	right=3pt,
    left=5pt,
	bottom=3pt},
}
\newkeytheoremstyle{tcbthm}{
tcolorbox-no-titlebar={ colframe=purple, interior hidden, breakable, before skip=10pt, after skip=10pt, outer arc=0pt,
	arc=0pt,
	colback=purple!3!white,
	top=0pt,
	right=0pt,
	bottom=0pt,
    },
}
\newkeytheorem{thm}[style=tcbthm,name=Theorem,parent=section]
\newkeytheorem{conj}[style=tcbconj,name=Conjecture,parent=section]
\newkeytheorem{theorem}[style=tcbext,name=Theorem,sibling=thm]
\newkeytheorem{lemma}[style=tcbext,sibling=thm]
\newkeytheorem{lem}[sibling=thm,name=Lemma]
\newkeytheorem{corollary}[style=tcbext,sibling=thm]

\zcsetup{countertype={thm=theorem}}
\zcsetup{countertype={cor=corollary}}
\zcsetup{countertype={lem=lemma}}
\zcsetup{countertype={ex=example}}
\zcsetup{countertype={ndefi=definition}}

\newkeytheoremstyle{tcbdefi}{tcolorbox-no-titlebar={colframe=cyan,interior hidden, breakable,before skip=10pt,after skip=10pt,	outer arc=0pt, 	arc=0pt, blanker,
borderline west={1pt}{0pt}{cyan},
top=3pt,
	right=3pt,
    left=5pt,
	bottom=3pt},
}
\newkeytheorem{defi}[style=tcbdefi,name=Definition,sibling=definition]
\newkeytheorem{nota}[sibling=defi,name=Notation]

\tcolorboxenvironment{ndefi}{
	colframe=blue,interior hidden, breakable,before skip=10pt,after skip=10pt,	outer arc=0pt,
	arc=0pt,
	colback=blue!3!white,
	rightrule=0pt,
	toprule=0pt,
	top=0pt,
	right=0pt,
	bottom=0pt,
	bottomrule=0pt,
}

\tcolorboxenvironment{intu}{
	colframe=green,interior hidden, breakable,before skip=10pt,after skip=10pt,	outer arc=0pt,
	arc=0pt,
	colback=blue!3!white,
	rightrule=0pt,
	toprule=0pt,
	top=0pt,
	right=0pt,
	bottom=0pt,
	bottomrule=0pt,
}

\tcolorboxenvironment{nota}{
	colframe=yellow,interior hidden, breakable,before skip=10pt,after skip=10pt,	outer arc=0pt,
	arc=0pt,
	colback=yellow!3!white,
	rightrule=0pt,
	toprule=0pt,
	top=0pt,
	right=0pt,
	bottom=0pt,
	bottomrule=0pt,
}

\tcolorboxenvironment{lem}{
	colframe=blue,interior hidden, breakable,before skip=10pt,after skip=10pt,	outer arc=0pt,
	arc=0pt,
	colback=white,
	top=0pt,
	right=0pt,
	bottom=0pt,
}
\newenvironment{clproof}{ \trivlist
	\item[\hskip\labelsep
	\emph{Proof of the claim}.]\ignorespaces
}{\hfill$\vartriangleleft$\vspace{10pt}\par}
\newcommand{\restr}[2]{\mathop{{#1}\mathbin{\downharpoonright}_{#2}}}

\usepackage{todonotes}
\newcommand{\Cc}{\mathscr C}
\newcommand{\Dd}{\mathscr D}
\newcommand{\Mm}{\mathscr M}
\newcommand{\Nn}{\mathscr N}
\newcommand{\struc}[1]{\mathbf{#1}}
\renewcommand{\phi}{\varphi}

\newcommand{\ndef}[2][]{\emph{#2}}
\begin{document}	

\begin{abstract}
Tree-ordered weakly sparse models have recently emerged as a robust framework for representing structures in an ``almost sparse'' way, while allowing the structure to be reconstructed through a simple first-order interpretation. 
A prominent example is given by twin-models, 
which are bounded twin-width tree-ordered weakly sparse representations of structures with bounded twin-width derived from contraction sequences. 
In this paper, we develop this perspective further.

First, we show that twin-models can be chosen such that they preserve linear clique-width or clique-width up to a constant factor. 

Then, we introduce \emph{merge-models}, a natural analog of twin-models for merge-width. 
Merge-models represent binary relational structures by tree-ordered weakly sparse structures. The original structures can then be recovered by a fixed first-order interpretation. A merge-model can be constructed from a merge sequence. Then, its
radius-$r$ merge-width will be, up to a constant factor, bounded by the radius-$r$ width of the merge sequence from which it is derived.

Finally, we show that twin-models arise naturally as special cases of merge-models, and that binary structures with bounded twin-width are exactly those having a loopless merge-model with bounded radius-$r_0$ merge-width (for some sufficiently large constant $r_0$).
\end{abstract}

\maketitle
\section{Introduction}

One of the central goals of current structural graph theory is to develop a common language for sparse and dense graph classes that still preserves the algorithmic and model-theoretic robustness familiar from sparsity theory.
One of the most general approaches and a particularly promising framework for this program is provided by \emph{tree-ordered weakly sparse models}: tree-like representations in which the additional unordered structure remains weakly sparse and where the modeled structure is obtained as a simple first-order interpretation of the model. 

\emph{Merge-width}~\cite{dreier2025merge} is a recently introduced width measure designed as a common generalization of several major graph parameters, including tree-width, twin-width, clique-width, degeneracy, and generalized coloring numbers.
This way, merge-width provides a unified perspective on two of the main structural frameworks in the area: sparsity theory and twin-width theory.
Its underlying decompositions, called \emph{merge sequences}, generalize the contraction sequences underlying the notion of twin-width, by keeping track of when adjacencies between parts become determined.

In this paper, we introduce the notion of \emph{merge-models} as a natural generalization  of the twin-models introduced in \cite{lmcs_perm} for twin-width.
They fit very naturally into the perspective of tree-ordered weakly sparse models. 
%
We now outline how  merge-models and their properties fit into this broader perspective.
\medskip

Recent work shows that one can capture a surprisingly rich fragment of monadically dependent structure theory and build new bridges to classical sparsity notions~\cite{tows_arxiv}. 
In particular, tree-ordered weakly sparse structures arise naturally for several important graph classes, including classes of bounded shrubdepth~\cite{ganian2019shrub,ganian2012trees}, structurally bounded expansion~\cite{gajarsky2020first}, monadically stable~\cite{covers}, bounded clique-width~\cite{courcelle1992monadic,colcombet2007combinatorial}, and bounded twin-width~\cite{twin-width1,lmcs_perm}; we defer the definitions needed in this paper to \zcref{sec:prelim}.

Among the most recent examples of this paradigm are \emph{twin-models}, which arise from contraction sequences witnessing bounded twin-width.
Introduced only recently, twin-width has quickly emerged as a major width measure, with numerous structural, algorithmic, and model-theoretic applications, see~\cite{twin-width1,bonnet2024twin}.
Twin-models provide a tree-ordered representation of graphs of bounded twin-width and played a key role in the characterization of bounded twin-width classes as first-order transductions of proper permutation classes~\cite{lmcs_perm}.
(See \zcref{sec:tww} for a formal definition of twin-models.)
The main building block behind this application is that every graph of twin-width at most~$t$ admits a weakly sparse twin-model of twin-width at most~$2t$. By weakly sparse, we mean that the biclique number\footnote{The \emph{biclique-number} $\bomega(\mathbf G)$ of a tree-ordered graph $\mathbf G$ is the maximum integer $t$ such that $K_{t,t}$ is a subgraph of the underlying graph. This notion naturally extends to tree-ordered binary structures, by considering the Gaifman graph of the structure obtained by ``forgetting'' the tree-order. (See \zcref{def:bomega}.)}~$\bomega$ of the constructed twin-model is bounded by a function of the twin-width of the represented graph.
Thus, twin-models witness how a decomposition parameter can be transformed into a tree-ordered weakly sparse representation that is stable under logical interpretation and useful for structural applications.

Classes of bounded linear clique-width and, more generally, of bounded clique-width form a natural and important source of examples of bounded twin-width classes.
Hence, they admit twin-models of bounded twin-width.
In this paper, we strengthen this observation by showing that one can choose these twin-models so that they preserve the corresponding clique-width parameter up to a constant factor: 
every binary relational structure of bounded linear clique-width admits a weakly sparse twin-model of bounded linear clique-width, and every binary relational  structure of bounded clique-width admits a weakly sparse twin-model of bounded clique-width, where the biclique number $\bomega$ of the constructed twin-models is bounded by a function of the linear clique-width (resp.\ the clique-width) of the represented graph.
More precisely, we prove the following theorem.

\getkeytheorem{thm:cw_model}

In this context we then study \emph{merge-width}~\cite{dreier2025merge}, a recently introduced width measure designed as a common generalization of several major graph parameters, including tree-width, twin-width, clique-width, degeneracy, and generalized coloring numbers.
Let us recall that in this way, merge-width provides a unified perspective on two of the main structural frameworks in the area: sparsity theory and twin-width theory.
Its underlying decompositions, called \emph{merge sequences}, refine contraction sequences by keeping track of when adjacencies between parts become determined.

\begin{figure}
    \centering
    \includegraphics[width=.2\textwidth]{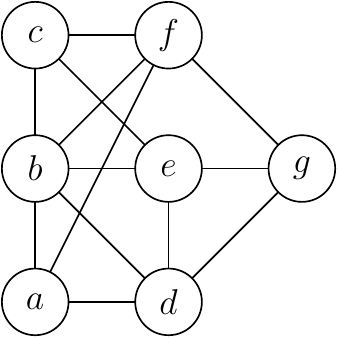}\\
    \hrulefill
    \vspace{12pt}\\\includegraphics[width=.75\linewidth]{mod_tww_gray}\\    
    \hrulefill
\vspace{12pt}\\
\includegraphics[width=.75\linewidth]{mod_mw_gray}    
    \caption{For the same graph,
     a contraction sequence and its associated twin-model, and a merge sequence
     and its associated (compact) merge-model (dark gray  plain edges ($S_{E,1}$) represent revealed adjacencies and dark gray dotted ones ($S_{E,0}$) revealed non-adjacencies).}
    \label{fig:seqmod}
\end{figure}

We further  introduce the notion of \emph{merge-models} as the natural generalization of twin-models for merge-width.
They fit naturally into the perspective of tree-ordered weakly sparse models.
Let us emphasize that just as twin-models encode the combinatorics of contraction sequences in a tree-ordered structure from which the graph can be recovered by a fixed interpretation, merge-models are 
tree-ordered structures that 
can be derived from merge sequences\footnote{Our definition of merge models also includes tree-ordered models that do not have a compatible ranking and thus do not derive from any merge sequence.} and serve, again, as canonical representations of the original graph.
In this way, merge-width is brought into the broader program of understanding graph classes through sparse structure together with an ordered tree-like skeleton.
\zcref{fig:seqmod} illustrates, for the same graph, a contraction sequence and its twin-model as well as a merge sequence and its merge-model.

Our main technical contribution is the following analogue of the above theorem for radius-$r$ merge-width.

\getkeytheorem{thm:mw_model}

Our results are primarily structural.
We isolate the notion of merge-model, develop its basic theory, and show that binary relational structures can be faithfully represented in this framework.
Although the resulting theory currently has fewer consequences than the one available for twin-models, it provides a clean structural foundation for studying merge-width through canonical tree-ordered representations.

We view this as another step in the emerging theory of tree-ordered weakly sparse models.
While the algorithmic and model-theoretic consequences of merge-models remain to be explored, the structural picture already suggests that they are the natural representation-theoretic counterpart of merge sequences, much as twin-models are for contraction sequences.
As further evidence for their relevance, we show that twin-models can be viewed naturally as simple special cases of merge-models.

\getkeytheorem{thm:mw_tww}

The remainder of the paper develops these ideas formally.
We recall the necessary preliminaries on graphs, relational structures, twin-models, and merge-width in \zcref{sec:prelim}.
We prove our transfer theorems for clique-width and merge-width in \zcref{sec:more}. 
In the remainder of the paper we introduce merge-models and establish their basic properties. 
\section{Preliminaries}
\label{sec:prelim}
\subsection{General notations}

We denote by $[k]$ the set $\{1,\ldots,k\}$.

A (finite) \emph{binary relational signature} $\sigma$ is a finite set of relation symbols, each of arity at most~$2$. 
Throughout the paper, all signatures are finite and binary. 

A (finite) \emph{$\sigma$-structure} $\struc M$ consists of a non-empty finite set $M$, called its \ndef{universe}, together with an interpretation $R(\struc M)\subseteq M^{\operatorname{arity}(R)}$ for every $R\in\sigma$.
We often do not distinguish between relation symbols and their interpretations. For $R\in\sigma$, we write $\mathbf M^R$ for the $R$-reduct of $\mathbf M$. More generally, for $\tau\subseteq\sigma$, we write $\mathbf M^\tau$ for the $\tau$-reduct of $\mathbf M$.

Let $\sigma$ be a relational signature.
The \emph{Gaifman graph}
${\rm Gaif}(\mathbf M)$ of a $\sigma$-structure $\mathbf M$ is the graph with same domain as $\mathbf M$, where $u$ and $v$ are adjacent if they belong together to some relation of $\mathbf M$.

A \emph{graph} is a finite structure over the signature $\{E\}$, where $E$ is interpreted as a symmetric and irreflexive binary relation.

A strict partial order $\prec$ is a \ndef{tree-order} if it has a unique minimum $r$ (the \ndef{root}) and, for every $x\neq r$, the set of elements smaller than $x$ forms a chain. Equivalently, the Hasse diagram of $(M,\prec)$ is a rooted tree. If $x\neq r$ is an element of a tree-order, we write $\pi(x)$ for the \emph{parent} of $x$, that is the largest $y$ such that $y\preceq x$ (which exists since we consider finite orders). Given two elements $x,y$ in a (partial) order, we write $x\parallel y$ if $x$ and $y$ are incomparable.

Let $\sigma$ be a relational signature.
A \emph{tree-ordered} $\sigma$-structure is a structure with signature $\sigma\cup\{\prec\}$, where $\prec$ is interpreted as a tree-order. An \emph{ordered} $\sigma$-structure is a special case of tree-ordered $\sigma$-structure when $\prec$ is a linear order (usually denoted~$<$ instead of $\prec$).

\begin{defi}
\label{def:bomega}
Let $\mathbf M$ be a tree-ordered $\sigma$-structure. The biclique number of $\mathbf M$ is defined as
\[
\bomega(\mathbf M)=\max\,\{t\colon K_{t,t}\subseteq \rm Gaif(\mathbf M^\sigma)\}.
\]
By extension, 
if  $\mathscr C$ is a class of tree-ordered  $\sigma$-structures, the biclique number of $\mathscr C$ is 
\[
\bomega(\mathscr C)=\sup\,\{\bomega(\mathbf M)\colon  \mathbf M\in\mathscr C\}.
\]
\end{defi}

\subsection{Logic}
A \emph{simple interpretation} $\mathsf I$ of $\tau$-structures in $\sigma$-structures consists  of a $\sigma$-formula $\nu(x)$ and a $\sigma$-formula $\rho_R(\bar x)$ for each  $R\in \tau$ (with arity~$|\bar x|$). For a $\sigma$-structure $\mathbf M$, 
$\mathsf I(\mathbf M)$ is the $\tau$-structure $\mathbf N$ with domain~$\nu(\mathbf M)$, such that $R(\mathbf N)=\{\bar a\subseteq\mathbf N\colon \mathbf M\models\rho_R(\bar a)\}$ for every $R\in\tau$.
For a class $\Mm$ of $\sigma$-structures we define $\mathsf I(\Mm)=\bigcup_{\struc M\in \Mm} \mathsf I(\struc M)$. 
We say that a class $\Nn$ can be interpreted in a class~$\Mm$ if there exists an interpretation~$\mathsf I$ such that $\Nn\subseteq \mathsf I(\Mm)$. 

Note that if $\tau\subseteq\sigma$, then the maps $\mathbf M\mapsto \mathbf M^\tau$ and $\mathbf M\mapsto {\rm Gaif}(\mathbf M)$ are both simple interpretations.
\medskip

For $k\in\mathbb N$, the \ndef{copy operation} $\mathsf C_k$ maps a $\sigma$-structure $\mathbf M$ to the $\sigma\cup\{E\}$-structure~$\mathsf C_k(\mathbf M)$ obtained by taking $k$ disjoint copies of $\mathbf M$ and making the clones of an element of $\mathbf M$ adjacent in $E$. (Note that $\mathsf C_1$ is the identity mapping.)

A unary predicate is also called a \emph{color}. 
A \ndef{monadic expansion} or \emph{coloring} of a \mbox{$\sigma$-structure~$\struc M$} is a \mbox{$\sigma^+$-structure~$\struc M^+$}, where $\sigma^+$
is obtained from $\sigma$ by adding unary relations, such 
that $\struc M$ is the $\sigma$-reduct of $\struc M^+$.
For a set $\Sigma$ of colors, the \ndef{coloring operation} $\Gamma_\Sigma$ maps a $\sigma$-structure $\struc M$ to the set~$\Gamma_\Sigma(\struc M)$ of all its $\Sigma$-colorings.

A \emph{transduction} $\mathsf T$ is a composition
of copy operations, monadic expansions, and simple interpretations.
Every transduction $\mathsf T$ is equivalent to the composition
$\mathsf I\circ\Gamma_\Sigma\circ \mathsf C_k$ of a copy operation $\mathsf C_k$, a
coloring operation~$\Gamma_\Sigma$, and a simple interpretation~$\mathsf I$ of $\tau$-structures in $\Sigma$-colored \mbox{$\sigma$-structures} \cite{SBE_TOCL}. 
Hence, for every $\sigma$-structure we have
$\mathsf T(\struc M)=\{\mathsf I(\struc M^+): \struc M^+\in\Gamma_\Sigma(\mathsf C_k(\struc M))\}$. 
(When we define a transduction, if we do not say anything about the copy operation, this means that it is not present, i.e.\ it is $\mathsf C_1$.)

For a class $\Mm$ of $\sigma$-structures we define $\mathsf T(\Mm)=\bigcup_{\struc M\in \Mm} \mathsf T(\struc M)$. 
We say that a class $\Nn$ can be transduced in a class~$\Mm$ if there exists a transduction $\mathsf T$ such that $\Nn\subseteq \mathsf T(\Mm)$. 

For a comprehensive discussion of first-order transductions, we refer the interested reader to \cite{Tquasiorder}.

\pagebreak
\subsection{Clique-width and linear clique-width}
The clique-width and linear clique-width parameters are dense versions of the tree-width and path-width parameters. Clique-width and linear clique-width allow to characterize when a graph (or, more generally, a binary structure) can be defined as a simple interpretation of  a colored tree 
(resp.\ of a colored path) in monadic second order logic. Classes with bounded clique-width have a particular importance in algorithmic graph theory, as it follows from \cite{Courcelle2000a,oum2006approximating} that every decision problem expressible in monadic second order logic is checkable in polynomial time on a class of graphs with bounded clique-width.

The formal definition of the clique-width and of the linear clique-width of a binary structure $\mathbf G$ is based on the complexity of a sequence of operations allowing to construct $\mathbf G$.

Let $\mathbf G$ be a  binary relational structure with finite signature $\sigma$.
The \emph{clique-width} of~$\mathbf G$, denoted by $\mathop{\rm cw}(\mathbf G)$, is  the minimum number of labels needed to construct $\mathbf G$ by means of the following $4$ operations:
\begin{itemize}
    \item create  a new element $v$ with label $i$;
    \item take the union of two structures previously constructed;
    \item for some relation $R\in\sigma$ and labels $i\neq j$, add all pairs $
    (u,v)$ to $R$ where $u$ has label $i$ and $v$ has label $j$;
    \item relabel all elements with label $i$ to label $j$.
\end{itemize}

The \emph{linear clique-width} of~$\mathbf G$, denoted by $\mathop{\rm lcw}(\mathbf G)$, is  
defined analogously as the minimum number of labels required to construct~$\mathbf G$ using the operations above, with the additional restriction that the construction is linear, that is, the disjoint union operation is only applied when one of the two structures consists of a single newly created element.

Clique-width and linear clique-width can alternatively be characterized at the level of classes,
from a logical point of view. A class of binary relational structures has bounded
clique-width if and only if it is a transduction of the class of tree-orders, and has bounded
linear clique-width if and only if it is a transduction of the class of linear orders
\cite{courcelle1995logical}. As a result, they are both preserved by transductions.

\begin{fact}\label{fact:cw_trans}
    Let $\sigma$ be a binary relational signature,
    and $\mathsf T$ a transduction from $\sigma$-structures. There exist functions $f_{\mathrm{cw}}$ and
    $f_{\mathrm{lcw}}$ such that for any $\sigma$-structure $\mathbf G$ and any $\mathbf H \in \mathsf T(\mathbf G)$,
    \[
    \mathrm{cw}(\mathbf H) \le f_{\mathrm{cw}}(\mathrm{cw}(\mathbf G)) \quad \text{ and }\quad
    \mathrm{lcw}(\mathbf H) \le f_{\mathrm{lcw}}(\mathrm{lcw}(\mathbf G)). 
    \]
\end{fact}

\subsection{Twin-width}
\label{sec:tww}
Twin-width is a parameter, which has been recently introduced in \cite{twin-width1}. Classes with bounded twin-width include both the classes with bounded clique-width and proper minor closed classes. 
Classes with bounded twin-width have strong algorithmic and structural properties. For instance, a class of ordered binary structures is monadically dependent if and only if it has bounded twin-width and, in such a case, FO-model checking is fixed-parameter tractable~\cite{bonnet2024twin}. 
This parameter is based on the notion of contraction sequences, from which one can construct tree-like representations, called \emph{twin-models} \cite{lmcs_perm}. 

Intuitively, a twin-model is a rooted tree that records a contraction history of the represented structure.
The leaves are the original vertices, while an internal node represents a part created during the contraction process.
A transversal edge~$Z_R(x,y)$ certifies that, in any stage in which the contractions represented by~$x$ and $y$ have already been realized, the relation~$R$ between the corresponding two parts is already determined.
The minimality condition says that we record each such piece of information at the earliest possible place in the tree, and the cycle condition prevents contradictory ways of propagating this information to the leaves. 
Formally, twin-models are defined as follows. 

\begin{defi}
Let $\sigma$ be a finite binary relational signature, and let \mbox{$\sigma'=\{\preceq\}\cup\{Z_{R}\colon R\in\sigma\}$}.
A \emph{twin-model} $\mathbf{T}$ is a $\sigma'$-structure such that the following properties hold in  $\mathbf{T}$:
\begin{itemize}
    \item  $\preceq$ is a tree-order, with minimal element $\rho(\mathbf T)$ (\/the \emph{root} of $\mathbf T$\/) and set of maximal elements $L(\mathbf T)$ (\/the set of \emph{leaves} of $\mathbf T$\/);

\item for any $Z_{R}\in\sigma'$, there is no pair of comparable elements $x$ and $y$ with $Z_R(x,y)$;

    \item for any $Z_{R}\in\sigma'$, each pair $(u,v) \in Z_{R}$ is minimal, in the sense that there is no distinct $(u',v') \in Z_{R}$ with $u' \preceq u$ and $v' \preceq v$;
    
    \item  any cycle  in $\mathbf T$ that follows the tree orientation (from root to leaves) must contain two consecutive edges in $\bigcup Z_{R}$.
\end{itemize}
\end{defi}

A twin-model $\mathbf T$ is a twin-model of a  $\sigma$-structure~$\mathbf G$ if $G=L(\mathbf T)$, and for  each relation $R\in \sigma$, $R(G)$ is the set of all pairs $(u,v)$ such that $u,v\in L(\mathbf T)$, and there exist $u' \preceq u$ and $v' \preceq v$ with $(u',v') \in Z_{R}$. 

Hence, the structure $\mathbf G$ is obtained as a simple interpretation $\mathsf I_\sigma$ (which depends only on $\sigma$) of any of its twin-models $\mathbf M$. It follows 
from \zcref{fact:cw_trans} and the preservation of twin-width boundedness by transductions \cite{twin-width1} that there exist function $f_{\rm lcw},f_{\rm cw}$, and  $f_{\rm tww}$ such that if $\mathbf M$ is a twin-model of $\mathbf G$, then 
${\rm lcw}(\mathbf G)\leq f_{\rm lcw}({\rm lcw}(\mathbf M))$, 
${\rm cw}(\mathbf G)\leq f_{\rm cw}({\rm cw}(\mathbf M))$, and
${\rm tww}(\mathbf G)\leq f_{\rm tww}({\rm tww}(\mathbf M))$.

Conversely, as mentioned in the introduction, $\mathbf M$ can be chosen so that 
${\rm tww}(\mathbf M)\leq 2\mathop{\rm tww}(\mathbf G)$.
We prove in \zcref{sec:more} that $\mathbf M$ can alternatively be chosen so that 
${\rm lcw}(\mathbf M)\leq 2\mathop{\rm lcw}(\mathbf G)$ or so that
${\rm cw}(\mathbf M)\leq 2\mathop{\rm cw}(\mathbf G)$.

\subsection{Merge-width}
\label{sec:mw}

Merge-width measures whether a graph admits a controlled hierarchical construction. 
Informally, vertices are gradually merged into larger parts, while pairs of vertices are progressively resolved as edges or non-edges. 
The key requirement is that this resolved information remains locally simple throughout the process. 
Formally, merge-width is defined via merge-sequences. 

Let $\sigma$ be a finite binary relational signature. We denote by $0$ and $1$ the false and true values. For $Z\in\sigma$ and $\alpha\in\{0,1\}$, we define 
\[
Z^\alpha(x,y):=(Z(x,y)\leftrightarrow\alpha)=\begin{cases}
    Z(x,y)&\text{if $\alpha=1$},\\
    \neg Z(x,y)&\text{otherwise}.
    \end{cases}
\]
and write $\bar\alpha = 1-\alpha$.

As in \cite{dreier2025merge}, for two sets $A,B$ we define
\[
AB:=\{(u,v)\in A\times B\colon  u\neq v\}.
\]
Unlike \cite{dreier2025merge}, we use the notation 
$\pairs{A}{2}$ for $\{(u,v)\in A\times A\colon  u\neq v\}$ (so that the falling factorial $(n)_2$ is the cardinality of $([n])_2$).

Let $\mathbf G$ be a $\sigma$-structure with domain $V$.

\begin{defi}
\label{defi:merge_seq}
A \emph{merge-sequence} of~$\mathbf G$
is  a sequence 
$$\Sigma=((\mathcal{P}_1,R_1),\ldots,(\mathcal{P}_m,R_m))$$ 
such that
\begin{enumerate}
    \item each $\mathcal{P}_i$ is a partition of $V$, $\mathcal P_{i-1}$ is a refinement of $\mathcal P_i$ (for $1<i\leq m)$,
    $\mathcal{P}_1=\{\{v\}:v\in V\}$, and $\mathcal{P}_m=\{V\}$;
    \item $R_1 \subseteq \dots \subseteq R_m=\pairs{V}{2}$;
    \item
   for every $i\in[m]$, every $Z\in\sigma$ and any two parts $P,Q\in \mathcal{P}_i$, there exists $\alpha\in\{0,1\}$ such that $PQ\setminus R_i \subseteq Z^\alpha$.
\end{enumerate}
\end{defi}

For $u,v\in V$ and $1\leq i\leq m$, we denote by ${\rm dist}_{R_i}(u,v)$ the shortest path distance between $u$ and $v$ in the underlying undirected graph $G(V,R_i)$. More generally, if $A\subseteq V$, we define ${\rm dist}_{R_i}(u,A)=\min_{v\in A} {\rm dist}_{R_i}(u,v)$.
We further define, for $1<i\leq m$ and $r$ a positive integer, the set $\reach_r(v,\mathcal{P}_{i-1},R_i)$ of all the parts in $\mathcal{P}_{i-1}$ accessible from $v$ by a path of length at most $r$ in the graph $(V,R_i)$, that is: 
\[
\reach_r(v,\mathcal{P}_{i-1},R_i)=\{P\in\mathcal P_{i-1}\colon {\rm dist}_{R_i}(v,P)\leq r\}.
\]
 
\begin{defi}
The \emph{radius-$r$ width}   of the  merge-sequence $\Sigma$, denoted by $\width_r(\Sigma)$, is defined as
$$\width_r(\Sigma)=\max_{2\le i \le m}\max_{v\in V}|\reach_r(v,\mathcal{P}_{i-1},R_i)|.$$
The  \emph{radius-$r$ merge-width} of $\mathbf G$, denoted by $\mw_r(\mathbf G)$, is defined as the minimum  radius-$r$ width among all merge-sequences of $\mathbf G$:
\[
\mw_r(\mathbf G)=\min_{\Sigma}\width_r(\Sigma).
\]
\end{defi}

\begin{fact}
\label{fact:min_sigma}
If $\Sigma=((\mathcal P_1,R_1),\dots,(\mathcal P_m,R_m))$ and $\Sigma'=((\mathcal P_1,R_1'),\dots,(\mathcal P_m,R_m'))$ are merge-sequences of $\mathbf G$ and $R_i'\subseteq R_i$ for all $1\leq i\leq m$, then $\width_r(\Sigma')\leq \width_r(\Sigma)$ for all positive integers $r$.
\end{fact} 

The following facts are immediate but worth mentioning.

\begin{fact}
\label{fact:mw_red}
    Let $\sigma$ be a binary relational signature, let $\tau\subset\sigma$, let $\mathbf G$ be a $\sigma$-structure, and let $r$ be a positive integer.
    Then,
\begin{align*}
{\rm mw}_r(\mathbf G^\tau)&\leq{\rm mw}_r(\mathbf G)\\
{\rm mw}_r({\rm Gaif}(\mathbf G))&\leq{\rm mw}_r(\mathbf G)
\end{align*}
\end{fact}
\begin{proof}
    Consider a merge sequence that is optimal for ${\rm mw}_r(\mathbf G)$ but reveal only the relations in $\tau$ (resp. the edges in the Gaifman graph of $\mathbf G$).
\end{proof}

\section{Linear clique-width and clique-width}
\label{sec:more}

\begin{thm}[store=thm:cw_model,restate-keys={note=see \zcpageref[nocap]{thm:cw_model}}]
 Let ${\mathbf G}$ be a  binary relational structure with finite signature $\sigma$.
 Then, ${\mathbf G}$ has a twin-model $\mathbf M$ with 
 \begin{align*}
 \bomega(\mathbf M)\leq 2\mathop{\rm lcw}(\mathbf G)\qquad&\text{and}\qquad  
 \mathop{\rm lcw}(\mathbf M)\leq  2\mathop{\rm lcw}(\mathbf G)
 \intertext{and a twin-model $\mathbf M$ with}
 \bomega(\mathbf M)\leq 2\mathop{\rm cw}(\mathbf G)\qquad&\text{and}\qquad  
 \mathop{\rm cw}(\mathbf M)\leq  2\mathop{\rm cw}(\mathbf G).
 \end{align*}
\end{thm}
\begin{proof}
      We consider a clique-expression
constructing $\mathbf G$ with $t$ labels, where $t=\mathop{\rm cw}(\mathbf G)$ if we consider general clique-expressions, or $t=\mathop{\rm lcw}(\mathbf G)$ if we consider linear clique-expressions.

Instead of directly constructing a twin-model of $\mathbf G$, we construct what we call a \emph{pre-model}, which is a twin-model where the root is deleted. (Note that, strictly speaking, the obtained structure is not tree-ordered anymore, but ``forest-ordered''.)

We construct a pre-model $\mathbf M$ following the clique-expression that 
constructs $\mathbf G$ as follows (the vertices of $\mathbf M$ are constructed with labels in $\{1,\ldots, t,1',\ldots, t'\}$ following a clique expression):
\begin{itemize}
    \item if the operation creates a new vertex $v$ with label $i$, we add a new (root) vertex of color $i$ to the current pre-model;
    \item if the operation combines the two substructures,  we take the union of the two corresponding  pre-models;
    \item if the operation adds the pairs $(u,v)$ to  $R\in \sigma$, where $u$ has label $i$ and~$v$ has label $j$, we first 
recolor every element colored $i$ and  $j$ to $i'$ and $j'$, respectively. 
We then add two new vertices $x$ and $y$ with colors $i$ and $j$,  extend the current model by adding the relations $x\prec x'$ for all elements $x$ of color $i'$, $y\prec y'$ for all elements $y'$ of color $j'$, and the transversal edge~$Z_R(x,y)$.
 \item 
 if the operation relabels all elements with label $i$ to label $j$, we first recolor all elements colored $i$ or $j$ to $j'$, and then add a new vertex $y$ with color $j$ and the relations $y\prec y'$ for all elements $y'$ of color $j'$.
    \end{itemize}
    
After performing all the above operations, we recolor every element of the resulting model with color $1'$, then introduce a new root $r$ with color $1$ and add $r\prec x$ for all element $x$ of color $1'$. 
Now, if there exist $x'\prec x$, $y'\prec y$ and $(x',y')\neq (x,y)$ such that $Z_R(x',y')$ and   $Z_R(x,y)$ for some $R\in \sigma$, then we delete $Z_R(x,y)$ from the current model.
This yields the final model $\mathbf M$.

We now prove that $\mathbf M$ is a twin-model. According to the construction of $\mathbf M$, the first three conditions of a twin-model hold. Assume that one can find  cycles  in $\mathbf T$ that follow the tree orientation (from root to leaves) and  contains no two consecutive edges in $\bigcup_{R\in \sigma} Z_R$. 
We consider a shortest cycle $\gamma$ with the above property.
By the minimality of $\gamma$,  $\gamma$ is an even cycle in which edges alternate between the tree order and relations $Z_R, R\in \sigma$.
Let  $\gamma=u_1,v_1,u_2,v_2,\ldots,u_p,v_p$, where $u_i\prec v_i$, and $Z_R(v_i,u_{i+1})$ or $Z_R(u_{i+1},v_i)$  (modulo $p$)  for some $R\in \sigma$.
We may assume that~$u_1$ is the vertex of the cycle $\gamma$ created  last during the model construction.
However, at the moment where $Z_R(u_1,v_p)$ is added, $u_p$ should not be present (by construction of $\mathbf M$, as $u_p\prec v_p$), which contradicts $u_1\prec u_p$.
Hence, the cycle $\gamma$ does not exist and thus  $\mathbf M$ is a twin-model. Clearly, the model  $\mathbf M$ has a clique-expression with $2t$ labels, which is linear if we started with a linear clique-expression.

Note that every vertex is originally colored with a color in $[t]$ and, from its first recoloring, will only get colors greater than $t$.
Moreover, transversal edges are added between vertices with colors in $[t]$. 
Let $Z=\bigvee_{R\in\sigma}Z_R$. 
As there exists no $u,u',v,v'$ such that $Z(u,v')$, $Z(u',v)$, $u'\prec u$ and $v'\prec v$, we deduce that if $K_{t,t}$ (with vertex sets $A$ and $B$) is a subgraph of the $Z$-reduct $M^Z$, then one side (say $B$) consists of incomparable vertices. 
Let $v\in A$ and let $u$ be a neighbor of $v$ in $B$ with, say, color~$i$. 
At the time the edge $uv$ is added, no neighbor of $v$ belongs to the same rooted tree as $u$, for otherwise it would be comparable with $u$.
As the maximum number of disjoint rooted trees at any stage of the construction is at most $2t$, we get $|B|\leq 2t$. Hence, $\bomega(\mathbf M)\leq 2t$.

It follows that  $\mathbf M$ satisfies  $\mathop{\rm cw}(\mathbf M)\leq  2\mathop{\rm cw}(\mathbf G)$
and $\bomega(\mathbf M)\leq  2\mathop{\rm cw}(\mathbf G)$
if it has been constructed from an optimal clique-expression of $\mathbf G$ and $\mathop{\rm lcw}(\mathbf M)\leq  2\mathop{\rm lcw}(\mathbf G)$ and 
$\bomega(\mathbf M)\leq  2\mathop{\rm lcw}(\mathbf G)$ if it has been constructed from an optimal linear clique-expression of $\mathbf G$.
\end{proof}

\section{merge-models}

We now define merge-models, which are generalizations of twin-models.
Again, the leaves are the vertices of the represented structure, and internal nodes encode larger parts arising along a decomposition.
The difference is that merge-models do not only record where adjacency is forced, but where adjacency or non-adjacency is decided.
In the following definition, an $S_{Z,1}$-edge means that the relation~$Z$ will hold by default between the corresponding two parts, while an $S_{Z,0}$-edge means that $Z$ does not hold by default.
In this setting, Condition~\textup{(5)} in~\Cref{def:merge_model} asserts that the status of every pair (in every relation) will eventually be decided. In the full generality of the definition given here, a merge-model does not need to arise from a merge sequence. However, as we shall see in the next section, every merge sequence uniquely defines a merge-model.

\begin{nota}
Let $\sigma^*$ be the  binary relational signature defined by
\[
\sigma^*=\{\preceq\}\cup\{S_{Z,\alpha}\colon Z\in\sigma, \alpha\in\{0,1\}\}.
\]
It will be convenient to define the following shortcuts: 
\begin{align*}
S_Z&=S_{Z,0}\vee S_{Z,1}&(Z\in\sigma)\\
S&=\bigvee_{Z\in\sigma}S_Z    
\intertext{and, for pairs $(x_1,y_1)$ and $(x_2,y_2)$,}
(x_1,y_1)\preceq(x_2,y_2)&=(x_1\preceq x_2)\wedge (y_1\preceq y_2).
\end{align*}
\end{nota}

\begin{ndefi}
\label{def:merge_model}
A \emph{merge-model} $\mathbf M$ is a $\sigma^*$-structure such that the following conditions hold in $\mathbf M$:
\begin{enumerate}
    \item\label{it:tree} $\preceq$ is a tree-order, with minimal element $\rho(\mathbf M)$ (\/the \emph{root} of $\mathbf M$\/) and set of maximal elements $V(\mathbf M)$ (\/the set of \emph{leaves} of $\mathbf M$\/); 
\item\label{it:nocomp} there is no pair of distinct comparable elements $x,y$ with $S(x,y)$:
\[
S(x,y)\rightarrow (x\parallel y)\vee(x=y);
\]
    \item\label{it:nocross} there exists no $(x,y)\prec (x',y')$ with
    $S(x,y')$ and $S(x',y)$; 
    \item\label{it:excl} there exists no $x,y$ and $Z\in\sigma$ with
    $S_{Z,0}(x,y)\wedge S_{Z,1}(x,y)$;
    \item\label{it:full} for every $Z\in\sigma$ and every  $(u,v)\in \pairs{V(\mathbf M)}{2}$ 
    there exists $(x,y)\preceq (u,v)$ with $S_Z(x,y)$.
\end{enumerate}
\end{ndefi}
\begin{ndefi}[Hat of a pair of leaves]
\label{def:hat}
According to Conditions (3) and (5) of \zcref{def:merge_model}, for every $(u,v)\in\pairs{V(\mathbf M)}{2}$ there exists a pair
$(x,y)\preceq (u,v)$ such that 
 $S(x,y)$ and $(x,y)$ is maximal for this property
--- meaning that $(x',y')\preceq (u,v)$ and $S(x',y')$ imply $(x',y')\preceq (x,y)$.

Hence, this pair $(x,y)$ is unique and we call it the \emph{hat} of $(u,v)$ and denote it by~$\widehat{u,v}$.
\end{ndefi}

A merge-model defines a $\sigma$-structure in the following way.

\begin{ndefi}[Interpretation $\Str$]
\label{nota:str}
We define the interpretation $\Str$ of $\sigma$-structures in $\sigma^*$-structures as follows: for a $\sigma^*$-structure $\mathbf M$,
\begin{itemize}
\item the domain of $\Str(\mathbf M)$ is $V(\mathbf M)$;
\item for every  $(u,v)\in \pairs{V(\mathbf M)}{2}$ and for every $Z\in\sigma$, we have
\[\Str(\mathbf M)\models Z(u,v)\quad\text{if}\quad \mathbf M\models S_{Z,1}(\widehat{u,v}).\]
\end{itemize}
\end{ndefi}

\section{Ranking and Layering}

We now relate merge-models to merge sequences.
The following construction $\model(\Sigma)$ turns a merge-sequence $\Sigma$ into a tree-ordered structure whose vertices are the parts appearing during the sequence.
The order is simply reverse inclusion of parts.
Whenever, at some step, all still-unrevealed pairs between two parts are known to be uniformly in $Z$ or uniformly outside $Z$, we add the corresponding transversal edge.
Thus $\model(\Sigma)$ stores exactly the information that the merge-sequence reveals, but in a static tree-ordered representation.

\pagebreak
\begin{nota}[Model derived from a sequence]
Let $\Sigma=((\mathcal P_1,R_1),\dots,(\mathcal P_m,R_m))$ be a merge-sequence of a $\sigma$-structure $\mathbf G$. The structure $\model(\Sigma)$ has signature~$\sigma^*$, domain $\bigcup \mathcal P_i$, with 
\begin{align*}
P\preceq Q\quad&\text{if}\quad P\supseteq Q\\
S_{Z,\alpha}(P,Q)\quad&\text{if}\quad\exists i<m, P,Q\in\mathcal P_{i},\text{ and }\emptyset\neq PQ\cap (R_{i+1}\setminus R_{i})\subseteq Z^\alpha(\mathbf G)
\end{align*}
(\/where $Z\in\sigma$ and $\alpha\in\{0,1\}$\/).

\end{nota}
\begin{lem}
\label{lem:seq_to_mod}
Let $\Sigma$ be a merge-sequence of a $\sigma$-structure $\mathbf G$. 

Then,
 $\model(\Sigma)$ is a merge-model and $\Str(\model(\Sigma))$ is isomorphic to $\mathbf G$.
\end{lem}
\begin{proof}
Let $\Omega$ be the domain of $\mathbf G$. (We use $\Omega$ instead of $V$ to avoid ambiguous notations).
Let $\Sigma=((\mathcal P_1,R_1),\dots,(\mathcal P_m,R_m))$. 

Let $f$ be the bijection between $\Omega$ and $\mathcal P_1=\{\{v\}\colon v\in \Omega\}$ that maps $v$ to $\{v\}$.

We first prove that $\model(\Sigma)$ is a merge-model:
\begin{itemize}
    \item That $\preceq$ is a tree-order follows from the fact that $\mathcal P_i$ is a refinement of $\mathcal P_j$ if $i<j$. The root of $\model(\Sigma)$ is $\rho(\model(\Sigma))=\Omega$ and its set of leaves is $V(\model(\Sigma))=\mathcal P_1=f(\Omega)$.
    \item By definition, if $\model(\Sigma)\models S_Z(P,Q)$, then there exists $i$ such that \mbox{$P,Q\in\mathcal P_i$}, hence either $P=Q$, or $P$ and $Q$ are disjoint
    and non-comparable with respect to $\preceq$.
    \item Assume $\model(\Sigma)\models S_Z(P,Q)\wedge S_Z(P',Q')$. Then, there exists $i,j$ with $P,Q\in\mathcal P_i$ and $P',Q'\in\mathcal P_j$. If $P\preceq Q'$, then $i\ge j$, and if $P'\preceq Q$, then $j \ge i$.
    Thus, we cannot have $P\preceq Q'$ and
$P'\preceq Q$ unless $i=j$, from which it follows that $(P',Q')=(P,Q)$.
    \item By construction and (3) of \zcref{defi:merge_seq}, we cannot have both $S_{Z,0}(P,Q)$ and $S_{Z,1}(P,Q)$ in $\model(\Sigma)$.
    \item As $R_m=\pairs{\Omega}{2}$, it follows that for every 
    pair $(\{u\},\{v\})$ of distinct leaves of $\model(\Sigma)$, there exists $i<m$ such that $(u,v)\in R_{i+1}\setminus R_i$.
Let $P\ni u, Q\ni v$ be such that $P,Q\in\mathcal P_i$. 
As $\Sigma$ is a merge sequence, there exists $\alpha\in\{0,1\}$
such that $PQ\setminus R_i\subseteq Z^\alpha$. Hence, 
by definition, $(P,Q)\preceq(\{u\},\{v\})$ and $S_Z(P,Q)$.
\end{itemize}

Thus, $\model(\Sigma)$ is a merge-model.
We now prove that the mapping $f$ is an isomorphism of $\mathbf G$ and $\Str(\model(\Sigma))$.

Note that the domain of $\Str(\model(\Sigma))$ is $f(\Omega)$. 
Moreover, for every $(u,v)\in \pairs{\Omega}{2}$ and every $Z\in\sigma$, we have: 
\begin{align*}
    \mathbf G\models Z(u,v)&\iff  \mathbf \Sigma\models\exists i\ (u,v)\in R_{i+1}\setminus R_{i}\\
    &\quad\text{ and }\exists P,Q\in\mathcal P_{i}\ 
    \bigl((u,v)\in PQ\cap (R_{i+1}\setminus R_{i})\subseteq Z\bigr)\\
    &\iff \mathbf \Sigma\models\exists i\ \exists P,Q\in\mathcal P_{i}\ 
    \bigl((u,v)\in PQ\cap (R_{i+1}\setminus R_{i})\subseteq Z\bigr)\\
    &\iff \model(\Sigma)\models (\exists (P,Q)\preceq (f(u),f(v))\ S_{Z,1}(P,Q))\\
    &\quad\text{ and }
    (\forall (P',Q')\preceq (f(u),f(v))\ (S_Z(P',Q')\rightarrow (P',Q')=(P,Q)))\\
    &\iff \model(\Sigma)\models S_{Z,1}(\widehat{f(u),f(v)})\\
    &\iff \Str(\model(\Sigma))\models Z(f(u),f(v)).
\end{align*}
 Hence, $f:\mathbf G\rightarrow \Str(\model(\Sigma))$ is an isomorphism.
\end{proof}

\pagebreak
\begin{ndefi}
\label{def:int_ranking}
An \emph{interval ranking} of a merge-model $\mathbf M$ is a mapping $\mathfrak I$ from $M$ to intervals of $\mathbb R$ such that:
\begin{enumerate}
\item if $x\prec y$, then $\mathfrak{I}(x)$ is to the right of $\mathfrak{I}(y)$ (i.e. $\max\mathfrak I(y)<\min\mathfrak I(x)$),
\item if $S(x,y)$, then $\mathfrak{I}(x)\cap\mathfrak{I}(y)\neq\emptyset$.
\end{enumerate}
A merge-model $\mathbf M$ equipped with an interval ranking  $\mathfrak I$ is a \emph{ranked merge-model}. 

A ranked merge-model $(\mathbf M,\mathfrak I)$ is \emph{clean} if there exists $m\in\mathbf N$ such that 
\begin{itemize}
    \item the set of the left endpoints of the intervals are all the integers in $[m]$,
    \item for every leaf $v$, $\min\mathfrak I(v)=1$;
    \item the interval associated to the root is $\mathfrak I(\rho(\mathbf M))=[m,m]$;
    \item for every non-root $v$, $\max\mathfrak I(v)=\min\mathfrak I(\pi(v))-1$.
\end{itemize}
\end{ndefi}

The interval ranking provides a linearized notion of ``time'' along the tree.
Leaves start at time~$1$, the root sits at the final time, and every internal node occupies exactly the interval during which the corresponding part exists in the merge process.

\begin{nota}
To every merge sequence $\Sigma=((\mathcal P_1,R_1),\dots,(\mathcal P_m,R_m))$ we associate  $\model^+(\Sigma):=(\mathbf M,\mathfrak I)$ as follows: $\mathbf M$ is the merge-model $\model(\Sigma)$ and~$\mathfrak I$ is the mapping from $M$ to intervals of $\mathbb R$ defined by
\[\mathfrak I(x)=[\min\{i\colon x\in\mathcal P_i\},\max\{i\colon x\in\mathcal P_i\}].\]
\end{nota}
\begin{lem}
\label{lem:seq_to_mod+}
    Let $\Sigma$ be a merge sequence of $\mathbf G$. Then, 
    $\model^+(\Sigma)$ is a clean ranked merge-model of $\mathbf G$.
\end{lem}
\begin{proof}
    That $\model^+(\Sigma)$ is a merge-model of $\mathbf G$ follows from \zcref{lem:seq_to_mod}. 
    That $\model^+(\Sigma)$ is clean is direct from the definition.
\end{proof}

A property, which is not used in this paper but is easily derived from \cite[Claim~4.7]{lmcs_perm}, is that a merge-model admits an interval ranking if and only if it satisfies the following consistency condition: if there exists $u_1,v_1,\dots,u_n,v_n$ with $u_i\preceq v_i$, $S(v_i,u_{i+1})$, and $S(v_n,u_1)$, then there exists $i\in[n]$ with $u_i=v_i$.

\begin{ndefi}[Cleaning] 
\label{def:cleaning}
The \emph{cleaning} of a ranked merge-model $(\mathbf M,\mathfrak I)$ is the pair $(\mathbf M,\mathfrak C)$, where $\mathfrak C$ is defined as follows:
Define
\begin{align*}
f(x)&:=\begin{cases}
    \min \mathfrak I(x)&\text{if $x\notin V(\mathbf M)$,}\\
    \min_{y\in V(\mathbf M)}\min \mathfrak I(y)&\text{otherwise;}
\end{cases}\\
g(x)&:=|\{f(y)\colon y\in\mathbf M, f(y)\leq f(x)\}|;\\
\mathfrak C(x)&:=\begin{cases}
[g(x),g(x)]&\text{if $x=\rho(\mathbf M)$,}\\
[g(x),g(\pi(x))-1]&\text{otherwise.}
\end{cases}
\end{align*}
\end{ndefi}

An example of a ranked merge model and its cleaning is given on \zcref{fig:clean}.

\begin{figure}[h!t]
    \centering
    \includegraphics[width=\linewidth]{clean}
    \caption{A ranked merge model $(\mathbf M,\mathfrak I)$ (for convenience, vertical segments have been added to the interval representation of $\mathfrak I$ to indicate the presence of $S$-relations) and a cleaning $(\mathbf M,\mathfrak C)$ of $(\mathbf M,\mathfrak I)$.} 
    \label{fig:clean}
\end{figure}

\begin{lem}
\label{lem:cleaning}
The cleaning $(\mathbf M,\mathfrak C)$ of a ranked merge-model 
$(\mathbf M,\mathfrak I)$ is a clean ranked merge-model (of $\Str(\mathbf M)$), which depends only on the relative order of the left endpoints of non-leaves of $\mathbf M$.

Moreover, if $(\mathbf M,\mathfrak I)$ is clean, then $\mathfrak C=\mathfrak I$.
\end{lem}
\begin{proof}
    For any non-root $x$, we have $f(x)<f(\pi(x))$, hence
    $g(x)<g(\pi(x))$. As~$g$ takes only integral values, we deduce $g(x)\leq g(\pi(x))-1$. It follows that $\mathfrak C$ is well-defined.

    We only have to check that $(\mathbf M,\mathfrak C)$ is a ranked merge-model, as it is then obviously clean.
    \begin{itemize}
        \item if $x=\pi(y)$, $\max \mathfrak C(y)=g(x)-1<\min \mathfrak C(x)$. Hence, $\mathfrak C(x)$ is to the right
        of $\mathfrak C(y)$;
        \item if $S(x,y)$, then $\mathfrak I(x)\cap \mathfrak I(y)\neq\emptyset$. Without loss of generality, we can assume $f(x)\leq f(y)\leq \max\mathfrak I(x)<f(\pi(x))$.
        It follows that $g(x)\leq g(y)<g(\pi(x))$. Hence,
        $g(y)\leq g(\pi(x))-1$. Thus, $\mathfrak C(x)\cap \mathfrak C(y)\neq\emptyset$.
    \end{itemize}

    The fact that the cleaning depends only on the relative order of the left endpoints of non-leaves of $\mathbf M$ follows from the following easy equality:
\[
g(x)=\begin{cases}
1&\text{if $x\in V(\mathbf M)$,}\\
    1+|\{\min\mathfrak I(y)\colon y\in\mathbf M\setminus V(\mathbf M), \min\mathfrak I(y)\leq \min\mathfrak I(x)\}|&\text{otherwise.}
\end{cases}
\]
The last property is direct from the definition of a clean ranked-merge-model.
\end{proof}

\begin{ndefi}
\label{def:merge_walk}
Let $(\mathbf M,\mathfrak I)$ be a ranked merge-model, let $\tau_\rho=\min\mathfrak I(\rho(\mathbf M))$, and
   let $\tau<\tau_\rho$.  
A \emph{$\tau$-bounded merge-walk} of order $n$ in $(\mathbf M,\mathfrak I)$  is a sequence
$(u_0,v_1,u_1,\dots,v_n,u_n,v_{n+1})$ where 
\begin{itemize}
	\item $u_{i-1}$ is comparable with $v_i$ in $\preceq$ (\/for $1\leq i\leq n+1$\/), (possibly $u_{i-1}=v_i$)
	\item $\mathbf M\models S(v_i,u_i)$ (\/for $1\leq i\leq n$\/),
    \item $u_0$ is a leaf, 
    \item $\max (\mathfrak I(u_i)\cap \mathfrak I(v_i))\leq\tau$,
    \item $\min\mathfrak I(v_{n+1})\leq\tau<\min \mathfrak I(\pi(v_{n+1}))$.
\end{itemize}
The set of the elements $v_{n+1}$ that can be reached from $u_0=v$ by a $\tau$-bounded merge walk in $\mathbf M$ of order at most $r$ is denoted by ${\rm MWReach}_r(\mathbf M,\mathfrak I,v,\tau)$ (or simply ${\rm MWReach}_r(v,\tau)$ when $\mathbf M$ and $\mathfrak I$ are clear from the context).
We define the $r$-width of $(\mathbf M,\mathfrak I)$ as 
\[
{\rm w}_r(\mathbf M,\mathfrak I)=\max_\tau\max_v|{\rm MWReach}_r(\mathbf M,\mathfrak I,v,\tau)|.
\] 
\end{ndefi}
A $\tau$-bounded merge-walk is a walk that only uses information already available by time~$\tau$.
This is why the parameter ${\rm w}_r(\mathbf M,\mathfrak I)$ matches the radius-$r$ width of the corresponding merge-sequence.

Remark that the inequality $\tau<\tau_\rho$ implies that no element of a $\tau$-bounded merge-walk is the root $\rho(\mathbf M)$.


\begin{lem}
\label{lem:ws_model}
    Let $(\mathbf M,\mathfrak I)$ be a ranked merge-model.
    
    Then, $\bomega(\mathbf M)\leq {\rm w}_1(\mathbf M,\mathfrak I)$.
\end{lem}
\begin{proof}
Note that the Gaifman graph of the unordered part of $\mathbf M$ is precisely $\mathbf M^S$.
Let $t:=\bomega(\mathbf M)$ and let $u_1,\dots,u_t,v_1,\dots,v_t$ be the endpoints of a $K_{t,t}$ in~$\mathbf M^S$.
Let
\begin{align*}
    a&:=\max_i\min \mathfrak I(u_i)\\
    b&:=\min_i\max\mathfrak I(u_i)\\
    \tau&:=\min_j\max\mathfrak I(v_j)
\end{align*}
By exchanging the $u_i's$ and the $v_j$'s we can assume $\tau\leq b$.

Note that $a\leq \tau$ (hence $a\leq b$), for otherwise the interval $\mathfrak I(u_i)$ with minimum~$a$ would not intersect the interval $\mathfrak I(v_j)$ with maximum $\tau$.
Let $v_j$ be such that $\max \mathfrak I(v_j)=\tau$ and let $x$ be a leaf descendant of $v_j$.
Then, for every $i\in [t]$, the sequence $(x,v_j,u_i,u_i)$ is a $\tau$-bounded merge-walk of order 1 as $\max\mathfrak I(v_j)\cap\mathfrak I(u_i)=\tau$ and
$\min\mathfrak I(u_i)\leq a\leq\tau$ and $\min\mathfrak I(\pi(u_i))>\max\mathfrak I(u_i)\geq b\geq\tau$.
Hence, \mbox{${\rm w}_1(\mathbf M,\mathfrak I)\geq t$}.
\end{proof}

\begin{lem}
\label{lem:clean_is_min}
    Let $(\mathbf M,\mathfrak I)$ be a ranked merge-model and let $\mathfrak C$ be the cleaning of~$\mathfrak I$.
    Then, ${\rm w}_r(\mathbf M,\mathfrak C)\leq {\rm w}_r(\mathbf M,\mathfrak I)$.
\end{lem}
\begin{proof}
Let $f$ and $g$ be defined as in \zcref{def:cleaning}.
Let $v\in V(\mathbf M), \tau_\rho=\min\mathfrak I(\rho(\mathbf M))$, $\tau\in[\tau_\rho-1]$, and let $(u_0=v,v_1,u_1,\dots,v_{n+1})$ be a $\tau$-bounded merge walk in $(\mathbf M,\mathfrak C)$. By construction, there exists $t,t^+\in\mathbf M$ such that $g(t)=\tau$ and $g(t^+)=\tau+1$. Let~$\tau'$ be just smaller than $f(t^+)$ (meaning that there is no endpoint of an interval $\mathfrak I(v)$ in the interval $[\tau',f(t^+)\,)$).

Then,
\begin{align*}
    \max(\mathfrak C(u_i)\cap\mathfrak C(v_i))\leq\tau&\iff 
    \min(g(\pi(u_i))-1,g(\pi(v_i))-1)\leq g(t)\\
    &\iff 
    \min(g(\pi(u_i)),g(\pi(v_i)))\leq g(t^+)\\
    &\iff 
    \min(f(\pi(u_i)),f(\pi(v_i)))\leq f(t^+)\\
    &\ \Longrightarrow 
    \min(\max\mathfrak I(u_i),\max\mathfrak I(v_i))<f(t^+)\\
    &\ \Longrightarrow 
    \max(\mathfrak I(u_i)\cap\mathfrak I(v_i))\leq \tau'
\end{align*}

Furthermore, 
\begin{align*}
    \min \mathfrak C(v_{n+1})\leq \tau <\min\mathfrak C(\pi(v_{n+1}))&\iff g(v_{n+1})\leq g(t)<g(t^+)\leq g(\pi(v_{n+1}))\\
    &\iff f(v_{n+1})\leq f(t)<f(t^+)\leq f(\pi(v_{n+1}))\\
    &\ \Longrightarrow \min \mathfrak I(v_{n+1})\leq
    \tau'<\min \mathfrak I(\pi(v_{n+1})).
\end{align*}

Hence, $(u_0=v,v_1,u_1,\dots,v_{n+1})$ is a $\tau'$-bounded merge walk in $(\mathbf M,\mathfrak I)$, from what follows that ${\rm w}_r(\mathbf M,\mathfrak C)\leq {\rm w}_r(\mathbf M,\mathfrak I)$.
\end{proof}

We now focus on the properties of clean ranked merge-models.

\begin{lem}
\label{lem:tau}
    Assume $(\mathbf M,\mathfrak I)$ is a clean ranked merge-model, $x$ is not the root, and $\tau$ is an integer.
    
    Then, the condition $\min\mathfrak I(x)\leq\tau<\min \mathfrak I(\pi(x))$ is equivalent to the condition
    $\tau\in \mathfrak I(x)$. 
\end{lem}
\begin{proof}
    Indeed, 
    \begin{align*}
        \tau\in \mathfrak I(x)&\iff \min\mathfrak I(x)\leq \tau\leq \max \mathfrak I(x)\\
        &\iff \min\mathfrak I(x)\leq\tau< \max \mathfrak I(x) +1 =\min \mathfrak I(\pi(x)).\qedhere
    \end{align*}
\end{proof}


\begin{nota}
To every clean ranked merge-model $(\mathbf M,\mathfrak I)$ we associate  a sequence $\Seq(\mathbf M,\mathfrak I)=((\mathcal P_1,R_1),\dots,(\mathcal P_{m+1},R_{m+1}))$ and antichains $\mathcal L_1,\dots\mathcal L_m$ of $\mathbf M$ 
as follows, where $m=\min\mathfrak I(\rho(\mathbf M))$.
For $i\in [m]$:
\begin{align*}
\mathcal L_i&:=\{x\colon i\in\mathfrak I(x)\},\\ 
\mathcal P_i&=\Bigl\{\{v\in V(\mathbf M)\colon v\succeq x\}\colon x\in\mathcal L_i\Bigr\},\\
R_i&=\Bigl\{(u,v)\in \tbinom{V(\mathbf M)}{2}\colon 
 \exists (x,y)\preceq (u,v)\ 
\bigl(S(x,y)\wedge\max (\mathfrak I(x)\cap \mathfrak I(y))<i \bigr)\Bigr\}.
\end{align*}
Moreover, $\mathcal P_{m+1}=\mathcal P_m$ and $R_{m+1}=\pairs{V(\mathbf M)}{2}$.
\end{nota}
\begin{lem}
\label{lem:clean_to_seq}
	For every clean ranked merge-model $(\mathbf M,\mathfrak I)$, the sequence $\Seq(\mathbf M,\mathfrak I)$ is a merge-sequence of $\Str(\mathbf M)$.
\end{lem}
\begin{proof}
Let 
$\Seq(\mathbf M,\mathfrak I)=((\mathcal P_1,R_1),\dots,(\mathcal P_{m+1},R_{m+1}))$, where $m=\min\mathfrak I(\rho(\mathbf M))$, and let $\mathcal L_i:=\{x\colon i\in\mathfrak I(x)\}$, for $i\in[m]$.
\begin{claim}
For all $i\in[m]$, 
$\mathcal L_i$ is a maximal antichain of $\preceq$.
\end{claim}
\begin{clproof}
The set $\mathcal L_i$ is an antichain as if $x\prec y$ then $\mathfrak I(x)\cap \mathfrak I(y)=\emptyset$.
Note that if $i=1$, then $\mathcal L_i=V(\mathbf M)$ is a maximal antichain. Now assume $i>1$.
Assume for contradiction that $\mathcal L_i$ is not a maximal antichain. Then, there exists $x\notin\mathcal L_i$ such that $x$ is incomparable in $\preceq$ with all the elements in $\mathcal L_i$. Let $x'\succeq x$ be a leaf of $\mathbf M$. Assume for contradiction that $x'$ is comparable (hence greater) than some element $y\in\mathcal L_i$. As $\preceq$ is a tree-order, we deduce from $x\preceq x'$ and $y\preceq x'$ that $x$ and $y$ are comparable, contradicting our assumption on $x$. 

From what precedes, we see that we can require that $x$ belongs to $\mathcal L_1=V(\mathbf M)$. Let $1\leq j\leq i$ be maximum such that 
$\mathcal L_j$ contains an element $x$ comparable to no element in $\mathcal L_i$. Of course, $j<i$ and $x\notin\mathcal L_{j+1}$. Thus, as the ranked merge-model $(\mathbf M,\mathfrak I)$ is clean, $\pi(x)\in\mathcal L_{j+1}$. As $\pi(x)$ is comparable with $x$, $\pi(x)\notin\mathcal L_i$, hence $j+1<i$.
By maximality of $j$, $\pi(x)$ is comparable with some element $y\in\mathcal L_i$. Because $\mathfrak I$ is an interval ranking, we have $y\preceq \pi(x)\preceq x$, a contradiction. 
\end{clproof}

By construction, the fact that $\mathcal L_i$ is a maximal antichain implies that $\mathcal P_i$ is a partition of $V(\mathbf M)$.
\begin{claim}
For all $m\geq j>i$ and all $x\in\mathcal L_i$ there exists $y\in\mathcal L_j$ with $y\preceq x$.
\end{claim}
\begin{clproof}

Let $j>i$ and $x\in\mathcal L_i$. As $\mathcal L_j$ is a maximal antichain, there exists $y\in \mathcal L_j$ that is comparable with $x$. As $i\in\mathfrak I(x)$ is less
than $j\in\mathfrak I(y)$, we deduce that~$\mathfrak I(x)$ is not to the right of $\mathfrak I(y)$, hence $x\not\prec y$. As $x$ and $y$ are comparable, $x\succeq y$.
\end{clproof}

This property implies (by construction) that if $j>i$ then $\mathcal P_i$ is a refinement of $\mathcal P_j$. Also, $\mathcal L_1=V(\mathbf M)$ and $\mathcal L_m=\{\rho(\mathbf M)\}$ implies $\mathcal P_1=\{\{v\}\colon v\in V(\mathbf M)\}$ and $\mathcal P_{m+1}=\{V(\mathbf M)\}$.

That $R_i\subseteq R_{i+1}$ follows immediately from the definition, and $R_{m+1}=\pairs{V(\mathbf M)}{2}$ by definition.

Let $i\in [m]$, $Z\in\sigma$, and $P,Q\in\mathcal P_i$.
Assume $R_i\neq PQ$. Let 
$(u,v)\in PQ\setminus R_i$, 
let $(x,y)=\widehat{u,v}$, and let $\alpha$ be such that $\mathbf M\models S_{Z,\alpha}(x,y)$. Then, $\Str(\mathbf M)\models Z^\alpha(u,v)$. As $(u,v)\notin R_i$, we deduce $\max(\mathfrak I(x)\cap \mathfrak I(y))\geq i$.
As $x\preceq u$ and $P\preceq u$ (and as $\prec$ is a tree-order), $x$ and $P$ are comparable.
As $i\in \mathfrak I(P)$ and $\max\mathfrak I(x)\geq i$, we deduce $x\preceq P$.
Similarly, $y\preceq Q$. Thus, $(x,y)\preceq (P,Q)$.
Assume $(x',y')\preceq (P,Q)$ with $\mathbf M\models S_Z(x',y')$. Then $(x',y')\preceq (u,v)$ with $\mathbf M\models S_Z(x',y')$. Hence, $(x',y')\preceq (x,y)$. 
 Hence, $(x,y)$ is independent of the choice of $u$ and $v$. It follows $PQ\setminus R_i\subseteq Z^\alpha$.
\end{proof}

\begin{lem}
	\label{lem:sigmap}
	Let $\Sigma$ be a merge-sequence of a $\sigma$-structure $\mathbf G$.
	
	Then,  for every positive integer $r$ we have ${\rm width}_r(\Seq(\model^+(\Sigma)))\leq {\rm width}_r(\Sigma)$.
	
\end{lem}
\begin{proof}
Let $\Sigma=((\mathcal P_1,R_1),\dots,(\mathcal P_m,R_m))$ and let $(\mathbf M,\mathfrak I)=\model^+(\Sigma)$.
By definition, the domain of $\mathbf M$ is $\bigcup_i\mathcal P_i$ and $i\in\mathfrak I(P)$ if and only if 
$P\in\mathcal P_i$. It follows that the merge sequence
$\Sigma'=\Seq(\mathbf M,\mathfrak I)$ is of the form
$\Sigma'=((\mathcal P_1,R_1'),\dots,(\mathcal P_m,R_m'),(\mathcal P_{m},R_m))$.
By definition, if $(u,v)\in R_i'$, then  there exists $(P,Q)\preceq (u,v)$ such that $S(P,Q)$ and $\max(\mathfrak I(P)\cap\mathfrak I(Q))<i$. 
By definition of $\model(\Sigma)$, the condition $S(P,Q)$ implies that there exists $j$ such that $P,Q\in\mathcal P_j$ and
$(u,v)\in PQ\cap (R_{j+1}\setminus R_j)$.
Hence, $j\in \mathfrak I(P)\cap \mathfrak I(Q)$. Thus, $j<i$. It follows that $(u,v)\in R_{j+1}\subseteq R_i$.

It is easily checked that if we extend the sequence $\Sigma$ by adding $(\mathcal P_{m+1},R_{m+1})$, where $\mathcal P_{m+1}=\mathcal P_m$ and $R_{m+1}=R_m$ we get a merge-sequence with the same ${\rm width}_r$.
Then, the difference between $\Sigma$ and $\Sigma'$ stands in the difference between $R_i$ (for $\Sigma$) and $R_i'$ (for $\Sigma'$). It is easily checked that $R_i'\subseteq R_i$. Hence, according to \zcref{fact:min_sigma}, we have ${\rm width}_r(\Sigma')\leq {\rm width}_r(\Sigma)$. 
\end{proof}

\begin{lem}
\label{lem:width=w}
    Let $(\mathbf M,\mathfrak I)$ be a clean ranked merge-model.
    Then, for every positive integer $r$, we have
\[
{\rm width}_r(\Seq(\mathbf M,\mathfrak I))={\rm w}_r(\mathbf M,\mathfrak I).
\]
\end{lem}
\begin{proof}
Let $m=\min \mathfrak I(\rho(\mathbf M))$, and let $((\mathcal P_1,R_1),\dots,(\mathcal P_{m+1},R_{m+1}))=\Seq(\mathbf M,\mathfrak I)$.
Let $1<i\leq m+1$, and let $\tau=i-1$.

Let $v\in V(\mathbf M)$ and $P\in {\rm MReach}_r(v,\mathcal P_{i-1},R_i)$. By definition, letting $t_0=v$, there exist
$t_1,\dots,t_n\in V(\mathbf M)$  such that
$n\leq r$, $(t_0,t_1,\dots,t_n)$ is a walk in the graph $(V(\mathbf M),R_i)$, and $t_n\in P$.
Let $u_0=t_0$.
For every $1\leq j< n$ we have $R_i(t_j,t_{j+1})$. Let $(v_j,u_j):=\widehat{t_j,t_{j+1}}$. Note that $\max(\mathfrak I(v_j)\cap\mathfrak I(u_j))<i$ (by the definition of $R_i$), hence $\max(\mathfrak I(v_j)\cap\mathfrak I(u_j))\leq \tau$.
By construction $u_j\preceq t_{j+1}$ and $v_{j+1}\preceq t_{j+1}$, hence~$u_j$ and $v_{j+1}$ are comparable.
Let $v_{n+1}:=P$. 
As $P\in \mathcal P_{i-1}$, we have $\tau\in\mathfrak I(v_{n+1})$. According to \zcref{lem:tau}, the sequence $(u_0,v_1,u_1,\dots,v_{n+1})$ is a $\tau$-bounded merge walk, hence ${\rm MReach}_r(v,\mathcal P_{i-1},R_i)\subseteq {\rm MWReach}_r(v,\tau)$.

Let $v\in V(\mathbf M)$ and $P\in {\rm MWReach}_r(v,\tau)$. 
By definition, there exists a $\tau$-bounded merge-walk $(u_0,v_1,u_1,\dots,v_n,u_n,v_{n+1})$ with $u_0=v$, $v_{n+1}=P$, and $n\leq r$.
Let $t_0:=v$. For $1\leq j\leq n$, let $t_j$ be a leaf descendant of both $u_j$ and~$v_{j+1}$. As $S(v_j,u_j)$,
$(v_j,u_j)\preceq (t_{j-1},t_j)$ and $\max(\mathfrak I(u_i)\cap \mathfrak I(v_j))\leq \tau<i$, we have $(t_{j-1},t_j)\in R_i$.
Hence, $(t_0,\dots,t_n)$ is a walk in the graph $(V(\mathbf M), R_i)$. Moreover, $P=v_{n+1}\preceq t_n$, hence $t_n\in P$ and $\tau\in\mathfrak{I}(v_{n+1})$ (by \zcref{lem:tau}) implies $P\in\mathcal P_{i-1}$. 
Thus, 
${\rm MWReach}_r(v,\tau)\subseteq {\rm MReach}_r(v,\mathcal P_{i-1},R_i)$.
\end{proof}
\begin{lem}
\label{thm:mw_w}
Let $\mathbf G$ be a $\sigma$-structure and let $r$ be a positive integer. Then
$\mw_r(\mathbf G)$ is the minimum of ${\rm w}_r(\mathbf M,\mathfrak I)$ over all (clean) ranked merge-models $(\mathbf M,\mathfrak I)$ of $\mathbf G$.
\end{lem}
\begin{proof}
Let $(\mathbf M,\mathfrak I)$  be a clean ranked merge-model of $\mathbf G$, which exists according to \zcref{lem:cleaning}.  According to \zcref{lem:width=w}, we have
${\rm width}
_r(\Seq(\mathbf M,\mathfrak I))={\rm w}_r(\mathbf M,\mathfrak I)$.
According to \zcref{lem:clean_to_seq}, 
$\Seq(\mathbf M,\mathfrak I)$ is a merge-sequence of 
$\Str(\mathbf M)$, which is isomorphic (by \zcref{lem:seq_to_mod}) to
$\mathbf G$.
Hence, ${\rm w}_r(\mathbf M,\mathfrak I)\geq \mw_r(\Str(\mathbf M))=\mw_r(\mathbf G)$.

Conversely,
let $\Sigma$ be a merge-sequence of $\mathbf G$ such that ${\rm width}_r(\Sigma)=\mw_r(\mathbf G)$, and let $(\mathbf M,\mathfrak I)=\model^+(\Sigma)$. Let $\Sigma'=\Seq(\model^+(\Sigma))$. According to \zcref{lem:seq_to_mod+}, $\model^+(\Sigma)$ is a clean ranked merge-model of $\mathbf G$ and,
according to \zcref{lem:sigmap}, ${\rm width}_r(\Sigma')\leq {\rm width}_r(\Sigma)=\mw_r(\mathbf G)$.

The fact that the minimum ${\rm w}_r(\mathbf M,\mathfrak I)$ over ranked merge-models $(\mathbf M,\mathfrak I)$ is attained by clean ranked merge-models follows from \zcref{lem:clean_is_min}.
\end{proof}
\section{Compactification}
We now define a notion of compactification of a (ranked) merge-model.
Compactification removes internal vertices that are independent with respect to
transversal edges ($S(\cdot,\cdot)$).
These vertices are still useful for the abstract tree structure, but they do not
affect the represented structure or the relevant width parameters.

\begin{nota}
Let $\mathbf M$ be a merge-model. The \emph{skeleton} $\Skel(\mathbf M)$ is the subset of (the domain of) $\mathbf M$ defined by
\[
\Skel(\mathbf M)=\{\rho(\mathbf M)\}\cup V(\mathbf M)\cup\{v\in\mathbf M\colon \exists x\in\mathbf M, S(x,v)\}.
\]
\end{nota}
\begin{ndefi}[compactification]~
\begin{trivlist}
\item Let $\mathbf M$ be a merge-model.
\item The \emph{compactification}  of $\mathbf M$ is the restriction $\mathbf M^c$ of $\mathbf M$ to $\Skel(\mathbf M)$.
\item
\item Let $(\mathbf M,\mathfrak I)$ be a ranked merge-model and let $\mathfrak C$ be
the restriction of $\mathfrak I$ to $\Skel(\mathbf M)$.
\item The \emph{compactification} of $(\mathbf M,\mathfrak I)$ is the pair $(\mathbf M^c,\mathfrak I^c)$, where $\mathbf M^c$ is the compactification of~$\mathbf M$ and  $(\mathbf M^c,\mathfrak I^c)$ is the cleaning of 
$(\mathbf M^c,\mathfrak C)$.
\item 
\item The ranked merge-model $(\mathbf M,\mathfrak I)$ is \emph{compact} if it is its own compactification. 
\end{trivlist}
\end{ndefi}

Note that the compactification of a ranked merge-model is compact and that every compact ranked merge-model is clean.


\begin{lem}
\label{lem:comp}
Let $(\mathbf M,\mathfrak I)$ be a ranked merge-model.  Then, $(\mathbf M^c,\mathfrak I^c)$ is a clean ranked merge-model of $\Str(\mathbf M)$.
\end{lem}
\begin{proof}
 First, note that $\Str(\mathbf M^c)=\Str(\mathbf M)$. Let $\mathfrak C$ be the restriction of $\mathfrak I$ to the domain of $\mathbf M^c$.
 Then, note that the properties characterizing an interval ranking are preserved by restriction: if $x\prec y$ and $x,y\in \mathbf M^c$, then $x\prec y$ in $\mathbf M$, thus $\mathfrak C(x)=\mathfrak I(x)$ is to the right of $\mathfrak C(y)=\mathfrak I(y)$. Similarly, if $x,y\in \mathbf M^c$ and $S(x,y)$ then $\mathfrak C(x)\cap\mathfrak C(y)=\mathfrak I(x)\cap\mathfrak I(y)\neq\emptyset$.
 Hence,  $(\mathbf M^c,\mathfrak C)$ is a ranked merge-model of $\Str(\mathbf M)$ and $(\mathbf M^c,\mathfrak I^c)$, obtained by cleaning $(\mathbf M^c,\mathfrak C)$ is a clean ranked merge-model of $\Str(\mathbf M)$ by \zcref{lem:cleaning}.
\end{proof}

\begin{nota}
Let $\mathbf G$ be a $\sigma$-structure, let $r$ be a positive integer, let $\mathscr R$ be the class of all ranked merge-models of $\mathbf G$, and 
let $(\mathbf N,\mathfrak C)$ be a compact ranked  merge-model.
We define
\[\widehat{\rm w}_r(\mathbf N,\mathfrak C)=\min\{{\rm w}_r(\mathbf M,\mathfrak I)\colon (\mathbf M,\mathfrak I)\in\mathscr R, \mathbf M^c=\mathbf N\text{ and }{\mathfrak I^c}=\mathfrak C\}.\]
\end{nota}
\begin{lem}
\label{lem:mw=hat}
Let $\mathbf G$ be a $\sigma$-structure and let $r$ be a positive integer. Then
\[\mw_r(\mathbf G)=\min\{\widehat{\rm w}_r(\mathbf N,\mathfrak C)\colon (\mathbf N,\mathfrak C)\text{ is a  compact ranked  merge-model of }\mathbf G\}.
\]
\end{lem}
\begin{proof}
Let $\mathscr R$ be the class of all ranked merge-models of $\mathbf G$  and let $\mathscr C$ be the class of all compact ranked merge-models of $\mathbf G$.

According to \zcref{thm:mw_w}, we have
    \begin{align*}
    \mw_r(\mathbf G)&=\min_{(\mathbf M,\mathfrak I)\in \mathscr R}  {\rm w}_r(\mathbf M,\mathfrak I)\\
    &=\min_{(\mathbf N,\mathfrak C)\in \mathscr C}\min\{
   {\rm w}_r(\mathbf M,\mathfrak I)\colon  (\mathbf M,\mathfrak I)\in\mathscr R, \mathbf M^c=\mathbf N\text{ and }{\mathfrak I^c}=\mathfrak C\}\\
   &=\min_{(\mathbf N,\mathfrak C)\in \mathscr C}\widehat{\rm w}_r(\mathbf N,\mathfrak C).\qedhere
    \end{align*}
\end{proof}

\begin{lem}
\label{lem:mhat=m}
	For every clean ranked merge-model $(\mathbf M,\mathfrak I)$ and every integer $r$ we have
\[
	{\rm w}_r(\mathbf M^c,\mathfrak{I}^c)\leq {\rm w}_r(\mathbf M,\mathfrak{I}).
\]
Hence,
\[
\widehat{\rm w}_r(\mathbf M^c,\mathfrak{I}^c)={\rm w}_r(\mathbf M^c,\mathfrak{I}^c)                             .
\]
\end{lem}
\begin{proof}
Let $\mathfrak C$ be the restriction of 
$\mathfrak I$ to $\Skel(\mathbf M)$.
As $(\mathbf M^c,\mathfrak I^c)$ is the cleaning of $(\mathbf M^c,\mathfrak{C})$, it follows from \zcref{lem:clean_is_min} that
${\rm w}_r(\mathbf M^c,\mathfrak{I}^c)\leq {\rm w}_r(\mathbf M^c,\mathfrak{C})$.
Moreover, 
it is clear from \zcref{def:merge_walk} that every $\tau$-bounded merge-walk of $(\mathbf M^c,\mathfrak C)$ is a $\tau$-bounded merge-walk of $(\mathbf M,\mathfrak I)$. Hence,
${\rm w}_r(\mathbf M^c,\mathfrak{C})\leq {\rm w}_r(\mathbf M,\mathfrak{I})$ and thus ${\rm w}_r(\mathbf M^c,\mathfrak{I}^c)\leq {\rm w}_r(\mathbf M,\mathfrak{I})$.

Then, the equality $\widehat{\rm w}_r(\mathbf M^c,\mathfrak{I}^c)={\rm w}_r(\mathbf M^c,\mathfrak{I}^c)$ follows from \zcref{lem:comp}.
\end{proof}

\begin{lem}
\label{lem:model of model}
	Let $({\mathbf M},\mathfrak I)$ be a compact ranked merge-model and let $r$ be an  integer.
	Then, ${\mathbf M}$ has  a compact ranked merge-model $(\mathbf N,\mathfrak{C})$ with
\[
	{\rm w}_r(\mathbf N,\mathfrak{C})\leq 2{\rm w}_r(\mathbf M,\mathfrak{I}).
\]
\end{lem}
\begin{figure}[h!t]
    \centering
    \includegraphics[width=\textwidth]{mod_mod_gray}
    \caption{Construction of a merge-model of a compact merge-model. On top, the original compact merge model $(\mathbf M,\mathfrak I)$ with signature $\{S_{E,0},S_{E,1},\prec\}$ encoding a graph.
    On the bottom, the constructed merge-model $(\mathbf N,\mathfrak C)$ of $\mathbf M$ with signature $\{S_{S_{E,0},0},S_{S_{E,0},1},S_{S_{E,1},0},S_{S_{E,1},1},S_{\prec,0},S_{\prec,1},\prec\}$.}
    \label{fig:mod_mod}
\end{figure}

\begin{proof}
Let $\sigma$ be the signature of $\mathbf M$ where the tree-order is denoted by $\prec$, and $\sigma' = \{S_{Z,\alpha} : Z\in \sigma, \alpha\in\{0,1\}\} \cup \{\prec\}$. (Here it will be convenient to consider that the signature contains the strict tree-order relation $\prec$ and not $\preceq$ to avoid to discuss loops at leaves of the model.)
	The domain of $\mathbf N$ is 
$\{\rho(\mathbf N)\}\cup \{(v,-1): v\in V(\mathbf M)\}\cup \{(u,i): u\in\mathbf M\setminus V(\mathbf M),i\in\{0,1\}\}$ and the relations are defined as follows (in the cases not considered, no relations hold; see \zcref{fig:mod_mod} for an illustration.):

\begin{itemize}
    \item for $\rho(\mathbf N)$:
\begin{align*}
\rho(\mathbf N)\text{ is the minimum of } \prec \text{ in }\mathbf N,\\
\mathbf N\models S_{Z,0} (\rho(\mathbf N),\rho(\mathbf N)),\quad Z\in \sigma,
\end{align*}

    \item for $(v,-1)$:
\begin{align*}
\mathbf N\models  (u,0)\prec (v,-1)\quad&\iff\quad 
\mathbf M\models u\prec v,\\
\mathbf N\models  S_{Z,1}((u,1),(v,-1))\quad&\iff\quad 
\mathbf M\models Z(u,v),\quad Z\in \sigma\setminus \{\prec\},\\
\mathbf N\models  S_{Z,1}((u,-1),(v,-1))\quad&\iff\quad 
\mathbf M\models Z(u,v),\quad Z\in \sigma\setminus \{\prec\},
\end{align*}

    \item for $(u,0)$ and $(v,j)$ with $j\in\{0,1\}$:
\begin{align*}
    \mathbf N\models (u,0)\prec (v,j)\quad&\iff\quad\mathbf M\models u\prec v\\
    \mathbf N\models S_{\prec,1}((v,1),(u,0))\quad&\iff\quad u=v
\end{align*}

    \item and for $(u,1), (v,1)$:
\begin{align*}
	\mathbf N\models S_{Z,1}((u,1),(v,1))\quad&\iff\quad \mathbf M \models Z(u,v),\quad Z\in \sigma\setminus\{\prec\}.
\end{align*}
\end{itemize}

Let $m=\min\mathfrak{I}(\rho(\mathbf M))$. We further define $\mathfrak{C}$ as follows:
    \begin{align*}
    \mathfrak{C}(\rho(\mathbf N))&=[m+1,m+1],\\
     \mathfrak{C}((v,-1))&=\mathfrak I(v),  \quad v\in V(M),\\
	\mathfrak{C}((u,i))&=\begin{cases}\mathfrak I(u)&\text{if }i\neq 1,\\
		[1,\max\mathfrak{I}(u)]&\text{otherwise.}
	\end{cases}
\end{align*}

We first prove that $(\mathbf N, \mathfrak C)$ is a ranked merge-model:
\begin{itemize}
    \item The tree-ordered set $\mathbf N^\prec$ can be seen as obtained from a copy of $M^\prec$ by adding a root ($\rho(\mathbf N)$) and leaves ($(M\setminus V(\mathbf M))\times \{1\}$). As $\mathbf M^\preceq$ is a tree-ordered set, it follows that $\mathbf M^\prec$ is also tree-ordered. The root of $\mathbf N$ is $\rho(\mathbf N)$ and its set of leaves is $  V(\mathbf N)= \{(v,-1): v\in V(\mathbf M)\}\cup
    \{(u,1):u\in \mathbf M \setminus V(\mathbf M)\}$.
    \item 
        By definition, if  $\mathbf N\models S((u,i),(v,j))$, then either $(u,i),(v,j)$ are leaves of $\mathbf N$, or  $u=v$ and $i\neq j$, or  $u\neq v$, $i=j=1$ and $\mathbf M\models S(u,v)$.
    In all cases, we have $(u,i)$ and $(v,j)$ are non-comparable with respect to $\preceq$ in $\mathbf N$.
    \item 
    Assume that there exists $((x,i),(y,j))\prec ((x',i'),(y',j'))$ in $\mathbf N$ such that
    $\mathbf N \models S((x,i),(y',j'))\wedge S((x',i'),(y,j))$. By construction, we must have  $i=j=0$, $(y',j')=(x,1)$ and $(x',i')=(y,1)$. It follows that $x\prec y$ and $y\prec x$ in $\mathbf M$, a contradiction.

    \item By construction, for any $x,y\in N$, we cannot have both $S_{Z,0}(x,y)$ and $S_{Z,1}(x,y)$ in $\mathbf N$.
    \item 
As $\mathbf N\models S_{Z,0} (\rho(\mathbf N),\rho(\mathbf N))$ for any $ Z\in \sigma$,
 for every $Z\in \sigma$ and every $(u,v)\in \pairs{V(\mathbf N)}{2}$
we can find $(x,y)\preceq (u,v)$ in $\mathbf N$ such that
    $\mathbf N \models S_Z(x,y)$.
    \item Finally, if $\mathbf N\models S((u,i),(v,j))$, then either $\mathbf M\models S(u,v)$ and thus $\mathfrak C((u,i))\cap \mathfrak C((v,j))\supseteq \mathfrak I(u)\cap\mathfrak I(v)\neq\emptyset$, or
    $u=v$ and $\mathfrak C((u,i))\cap \mathfrak C((u,j))\neq\emptyset$.
    Moreover, $\mathbf N\models (u,i)\prec (v,j)$ implies $i=0$, thus 
    $\mathbf M\models u\prec v$. Hence, $\max\mathfrak C((v,j))=\max\mathfrak I(v)<\min\mathfrak I(u)=\min\mathfrak C((u,i))$. 
\end{itemize}
It follows that $(\mathbf N, \mathfrak C)$
is  a ranked merge-model in which every internal node is adjacent to at least one $S$-edge.
As  $({\mathbf M},\mathfrak I)$  is clean by \zcref{lem:comp}, 
by construction, $(\mathbf N, \mathfrak C)$
is also clean.
Note that the leaves of $\mathbf N$ are in bijection with the domain of $\mathbf M$ by $(u,i)\mapsto u$.
It is then easy to check that, for every  $((u,i),(v,j))\in \pairs{V(\mathbf N)}{2}$ and for every $Z\in\sigma$, $\mathbf N\models S_{Z,1}(\widehat{(u,i),(v,j)})$, if and only if  $\mathbf M\models Z(u,v)$. It follows that $\Str(\mathbf N)=\mathbf M$. 
Hence, $(\mathbf N, \mathfrak C)$
is a compact ranked merge-model of  $(\mathbf M, \mathfrak I)$.

Let $(a_0,p_0)\in V(\mathbf N)$, let $\tau\in[m-1]$, and let \[((a_0,p_0),(b_1,q_1),(a_1,p_1),(b_2,q_2),\dots,(a_n,p_n),(b_{n+1},q_{n+1}))\] be a $\tau$-bounded merge-walk of length at most $r$ in $(\mathbf N,\mathfrak{C})$.

Let $u_0$ be any leaf of $\mathbf M$ such that $u_0\succeq a_0$ in $\mathbf M$.
Then,  $(u_0,b_1,a_1,\dots,a_n,b_{n+1})$ is a $\tau$-bounded merge-walk in $\mathbf M$:
for every integer $i\in[n]$, we have $\mathbf N\models S((a_i,p_i),(b_i,q_i))$, hence either $\mathbf M\models S(a_i,b_i)$ or $a_i=b_i$ and $p_i\neq q_i$ (by abuse of notation, we use the same notation $S$ in two different contexts).
In the later case, we can reduce the walk $(u_0,b_1,a_1,\dots,a_n,b_{n+1})$  by removing both $a_i$ and $b_i$ from the walk if the previous and next in the sequence are different from $a_i$:
$x,a,a,y\rightarrow x,y$.
It is easy to check that $x$ and $y$ are comparable in $\mathbf M$ with $\prec$.
Otherwise, we replace $a,a,a$ by 
$a,x,a$ where $\mathbf M\models S(a,x)$. Note that such an $x$ exists as $\mathbf M$ is compact.
Note also that $\max (\mathfrak I(a)\cap \mathfrak I(x))=\max (\mathfrak C((a,i))\cap \mathfrak C((x,j))\le \max \mathfrak C((a,i))
=\max (\mathfrak C((a,0))\cap \mathfrak C((a,1)))\le \tau$.
Hence, $b_{n+1}$ can always be reached by a $\tau$-bounded merge-walk with length at most $r$ in $\mathbf M$ from $u_0$.

It follows that the mapping $(x,i)\mapsto x$  maps ${\rm MWReach}_r(\mathbf N,\mathfrak{C},(v,i),\tau)$ to ${\rm MWReach}_r(\mathbf M,\mathfrak{I},u,\tau)$ (where $u$ is any leaf descendant of $v$ in $\mathbf M$). 
As is this mapping is $2$-to-$1$, ${\rm w}_r(\mathbf N,\mathfrak{C})\leq 2{\rm w}_r(\mathbf M,\mathfrak{I})$.
\end{proof}

\begin{thm}[store=thm:mw_model,restate-keys={note=see \zcpageref[nocap]{thm:mw_model}}]
 Let ${\mathbf G}$ be a  binary relational structure and let $r$ be a positive integer.
 Then, ${\mathbf G}$ has a compact merge-model $\mathbf M$ with 
 \[
 \bomega(\mathbf M)\leq {\rm mw}_r(\mathbf G)\qquad\text{and}\qquad  
 {\rm mw}_r(\mathbf M)\leq  2{\rm mw}_r(\mathbf G).
 \]
\end{thm}
\begin{proof}
    According to \zcref{lem:comp,lem:mw=hat,lem:mhat=m}, there exists a compact ranked merge-model $(\mathbf M,\mathfrak I)$ such that 
    \[{\rm mw}_r(\mathbf G)=\widehat{\rm w}_r(\mathbf M,\mathfrak I)=\rm w_r(\mathbf M,\mathfrak I).\]
    According to \zcref{lem:ws_model}, we have
\[
\bomega(\mathbf M)\leq {\rm w}_1(\mathbf M,\mathfrak I)\leq {\rm w}_r(\mathbf M,\mathfrak I)={\rm mw}_r(\mathbf G).
\]
    Moreover, by \zcref{lem:model of model}, $\mathbf M$ has a compact ranked merge-model $(\mathbf N,\mathfrak{C})$ such that
\[
	{\rm w}_r(\mathbf N,\mathfrak{C})\leq 2{\rm w}_r(\mathbf M,\mathfrak{I}).\]
Finally, according to \zcref{thm:mw_w},  we have 
\[
{\rm mw}_r(\mathbf M)\le {\rm w}_r(\mathbf N,\mathfrak{C})\leq 2{\rm w}_r(\mathbf M,\mathfrak{I})=2{\rm mw}_r(\mathbf G).\]
\end{proof}


\section{Merge-width and twin-width}
Note that twin-models are special types of merge-models (by possibly adding a loop at the root), and that every class of binary structures with bounded twin-width has bounded merge-width.
The aim of this section is to prove that having bounded twin-width is equivalent to having a loopless merge-model (for a sufficiently high radius). The precise statement of this property is stated in \zcref{thm:mw_tww}.

In order to prove this theorem, we rely on several results proved in the literature, which refer to graph invariants like the \emph{flip-width} ${\rm fw}$, the \emph{star-chromatic number} $\chi_s$, the \emph{grad} $\nabla_1$, and the (strong) \emph{generalized coloring number} ${\rm col}_2$. However, we won't need the specific definitions of these invariants and thus, we will omit them.

\begin{lemma}[Special case of {\cite[Theorem 6]{bonnet2024twin}}]
\label{lem:twwo}
    Let $\mathscr C$ be a hereditary class of ordered graphs. Then, the following are equivalent:
    \begin{enumerate}
        \item $\mathscr C$ has unbounded twin-width;
        \item $\mathscr C$ includes one of  $25$ special hereditary classes $\mathscr F_1,\dots,\mathscr F_{25}$;
        \item $\mathscr C$ interprets the class of all graphs.
    \end{enumerate}
\end{lemma}

\begin{lem}
\label{lem:tww_uniqueT}
    There exists a transduction $\mathsf T_{\rm univ}$ such that a class
     $\mathscr C$ of ordered graphs has bounded twin-width if and only if $\mathsf T_{\rm univ}(\mathscr C)$ is not the class of all graphs.
\end{lem}
\begin{proof}
We consider the classes $\mathscr F_1,\dots,\mathscr F_{25}$ witnessing non-boundedness of twin-width (see \zcref{lem:twwo}).
    Each class $\mathscr F_i$ interprets the class of all graphs. Thus, there are interpretations $\mathsf I_1,\dots,\mathsf I_{25}$ such that every $\mathsf I_i(\mathscr F_i)$ is the class of all graphs. Let~$\mathsf H$ be a transduction 
    such that  $\mathsf H(G)$ is the set of all induced subgraphs of $G$. 
    Let~$\mathsf I_i$ be defined by formulas $\nu_i(x)$ and $\rho_i(x,y)$. We define the transduction~$\mathsf T_0$ using formulas $\nu$ and $\rho$, where
    \begin{align*}
        \nu(x)&:=\bigvee_{i=1}^{25} \biggl(\bigl(\forall v\ P_i(v)\bigr)\wedge \nu_i(x)\biggr),\\                \rho(x,y)&:=\bigvee_{i=1}^{25} \biggl(\bigl(\forall v\ P_i(v)\bigr)\wedge \rho_i(x,y)\biggr),
    \end{align*}
    where $P_1,\dots,P_{25}$ are unary predicates.
    Let $\mathsf T_{\rm univ}=\mathsf T_0\circ\mathsf H$.
    Then, note that $\bigcup_{i=1}^{25}\mathsf I_i\circ\mathsf H(\mathscr C)\subseteq \mathsf T_{\rm univ}(\mathscr C)$.

    Thus, if $\mathscr C$ has unbounded twin-width, $\mathsf H(\mathscr C)$ contains one of the classes $\mathscr F_i$ and $\mathsf T_{\rm univ}(\mathscr C)$ is the class of all graphs.
    Reciprocally, if $\mathscr C$ has bounded twin-width, then so does $\mathsf T_{\rm univ}(\mathscr C)$. Hence, $\mathsf T_{\rm univ}(\mathscr C)$ is  not the class of all graphs.
\end{proof}

\begin{lemma}[\cite{lmcs_perm}]
\label{lem:Gaif}

Let $t\in\mathbb N$ and let $\tau,\tau'$ be binary relational signatures.
Let $\mathscr C$ be a class of $\tau$-structures, such that 
${\rm Gaif}(\mathscr C)$ has star chromatic number at most $t$.

Then there exists a quantifier-free transduction $\mathsf T_{\tau,\tau',t}$ such that for every $\mathbf M\in\mathscr C$, $\mathsf T_{\tau,\tau',t} (\mathbf M)$ contains all the $\tau'$-structures $\mathbf M'$ with  ${\rm Gaif}(\mathbf M')={\rm Gaif}(\mathbf M)$.
\end{lemma}

\begin{lemma}[{\cite[Lemma 6.1]{dreier2025merge}}]
\label{lem:mw_continuous}
    For every finite binary signature $\sigma$ and 
    every transduction $\mathsf T$ there exists a function $F_{\mathsf T}$ and an integer $c({\mathsf T})$ (depending only on the maximum quantifier rank of a formula used to define $\mathsf T$) such that for every $\sigma$-structure $\mathbf M$ and integer $r$
    \begin{equation*}
        {\rm mw}_r(\mathsf T(\mathbf M))\leq F_{\mathsf T}({\rm mw}_{c({\mathsf T})\cdot r}(\mathbf M)).
    \end{equation*}
\end{lemma}

\begin{lem}
\label{lem:star}
 Every weakly sparse class of graphs with bounded ${\rm mw}_3$ has bounded star chromatic number.
\end{lem}
\begin{proof}
    According to \cite[Theorem 7.6]{dreier2025merge}, we have
    ${\rm fw}_2(G)\leq 4^{{\rm mw}_{3}(G)}$.
    According to \cite[Section 6]{torunczyk2023flip}, we have
    $\nabla_1(G)\in O((\mathop{\rm fw_2(G))^{72})}$.
    Moreover, the star chromatic number $\chi_s(G)$ of a graph $G$ is bounded by its generalized coloring number ${\rm col}_2(G)$, 
    which is bounded by a function of $\nabla_1(G)$ \cite{Zhu2008}.
\end{proof}

\begin{thm}[store=thm:mw_tww,restate-keys={note=see \zcpageref[nocap]{thm:mw_tww}}]
    There exists an integer $r_0$ with the following property.
    Let $\mathscr C$ be a class of graphs. Then, the following are equivalent:
\begin{enumerate}
    \item $\mathscr C$ has bounded twin-width;
    \item $\forall r$, $\mathscr C$ has  loopless merge-models with  bounded ${\rm mw}_{r}$ and  $\bomega$;        \item $\mathscr C$ has loopless merge-models with bounded ${\rm mw}_{r_0}$ and  $\bomega$.
\end{enumerate}    
\end{thm}
\begin{proof}
For $(1)\Rightarrow (2)$:
let $\sigma$ be the signature $\{E,F\}$. By an interpretation we expand each $G\in\mathscr C$ into a $\sigma$-structure $G^+$ by 
defining $F(x,y)$ whenever $x$ and $y$ are not adjacent. The class of obtained $\sigma$-structures has bounded twin-width. For each $G\in\mathscr C$, a twin-model of $G^+$ naturally defines a loopless merge-model of $G$ (by $S_E\rightarrow S_E^1$ and $S_F\rightarrow S_E^0$). The bound on the biclique-number of the merge-models follows from the one on the twin-models.

The implication $(2)\Rightarrow (3)$ is trivial.

We now prove the implication $(3)\Rightarrow (1)$.

        According to \zcref{lem:mw_continuous}, for every binary relational signature $\sigma$ and every
        transduction $\mathsf T$ of  $\sigma$-structures there exists a function $F_{{\mathsf T}}$ and an integer $c({\mathsf T})$ such that for every $\sigma$-structure $\mathbf M$ we have ${\rm mw}_r({\mathsf T}(\mathbf M))\leq F_{{\mathsf T}}({\rm mw}_{c({\mathsf T})\cdot r}(\mathbf M))$.
        As~$c(\mathsf T)$ depends only on the maximum quantifier rank of a formula used to define~$\mathsf T$, we can define $c_0$ as the value of $c(\mathsf T)$ for quantifier-free transductions.

Let $r_0=\max(c_0\cdot\mathop{c}({\Str})\cdot\mathop{c}({\mathsf T}_{\rm univ}),3)$.

Let $k,t\in\mathbb N$, and  assume that each $G\in\mathscr C$ has a loopless merge-model $\mathbf M_G$ with  ${\rm mw}_{r_0}(\mathbf M_G)\leq k$ and $\bomega(\mathbf M_G)\leq t$.

Let $\sigma=\{E,<\}$ be the signature of ordered graphs, let $\tau=\{S_E^0,S_E^1\}$ and $\tau'=\tau\cup\{S_<^0,S_<^1\}$, the signature of $\mathbf M_G$.
According to \zcref{fact:mw_red}, ${\rm mw}_{r_0}({\rm Gaif}(\mathbf M_G^\tau))\leq {\rm mw}_{r_0}(\mathbf M_G)$. Hence, as $r_0\geq 3$, it follows from \zcref{lem:star} that  ${\rm Gaif}(\mathbf M_G^\tau)$ has star-chromatic number bounded by some constant $t'$ depending on $k$ and $t$.
It follows from \zcref{lem:Gaif} that there exists a  quantifier-free transduction $\mathsf X_{t'}$ from ordered $\tau$-structures to ordered  $\tau'$-structures leaving the linear order  unchanged and such that for every ordered $\sigma$-structure $\mathbf M$, if the star chromatic number of ${\rm Gaif}(\mathbf M)$ is at most $t'$, then $\mathsf X_{t'}(\mathbf M)^{\tau'}$ contains all the $\tau'$-structures with same Gaifman graph as $\mathbf M^\tau$.

For $G\in\mathscr C$ with merge-model $\mathbf M_G$, we consider a linear order $L$ on $\mathbf M_G$ obtained from a (DFS) traversal of the cover graph of the tree-order $\mathbf M_G^\prec$. Let $G_<$ be the expansion of $G$ into an ordered graph, where the linear order follows from the restriction of $L$ to the leaves of $\mathbf M_G$. We denote by $\Dd$ the obtained class of ordered graphs.

A merge-model $\mathbf M_{G_<}$ of $G_<$ is easily derived from the merge-model $\mathbf M_G$ of $G$ and the linear order $L$:
Let $\mathbf M_{G_<}$ be obtained from $\mathbf M_G$ by adding, 
for every nodes $u,v$ of $\mathbf M_G$ with 
$\mathbf M_G\models S(u,v)\wedge L(u,v)$ the relations $S_<^1(u,v)$ and $S_<^0(v,u)$ to~$\mathbf M_{G_<}$.
Assume $u<v$ in $G_<$. 
As $\mathbf M_G$ is a loopless merge-model of $G$, the hat $(u',v'):=\widehat{u,v}$ (in $\mathbf M_G$) is such that $u'\neq v'$. 
By construction, $(u',v')$ is also $\widehat{u,v}$ in ~$\mathbf M_{G_<}$
and 
$\mathbf M_{G_<}\models S_<^1(u',v')\wedge S_<^0(v',u')$. 
Thus,
$\Str(\mathbf M_{G_<})\models u<v$.
It follows that $\Str(\mathbf M_{G_<})=G_<$.

Let $H\in\mathsf T_{\rm univ}(G_<)$. We have the following situation:

\[
\begin{xy}
\xymatrix{\mathbf M_G\ar[r]^{\mathsf X_{t'}}\ar[d]_{\mathsf{Str}}&\mathbf M_{G_<}\ar[d]^{\mathsf{Str}}\\
G&G_<\ar[l]_{\mathrm{reduct}}\ar[r]^{\mathsf T_{\rm univ}}&H
}
\end{xy}
\]

Hence, as $c_0=c(\mathsf X_{t'})$, $r_0\ge c_0\cdot\mathop{c}(\Str)\cdot\mathop{c}(\mathsf T_{\rm univ})$, and ${\rm mw}_{r_0}(\mathbf M_G)\leq k$, the following inequalities follow from \zcref{lem:mw_continuous}:
\begin{align*}
{\mw}_{1}(H)&\leq F_{\mathsf T_{\rm univ}}({\rm mw}_{c(\mathsf T_{\rm univ})}(G_<))\\
&\leq (F_{\mathsf T_{\rm univ}}\circ F_{\Str})({\rm mw}_{\mathop{c}(\Str)\cdot\mathop{c}(\mathsf T_{\rm univ})}(\mathbf M_{G_<}))&\text{(as $G_<\in\Str(\mathbf M_{G_<})$)}\\
&\leq (F_{\mathsf T_{\rm univ}}\circ F_{\Str}\circ F_{\mathsf X_{t'}})({\rm mw}_{c_0\cdot\mathop{c}(\Str)\cdot\mathop{c}(\mathsf T_{\rm univ})}(\mathbf M_G))&\text{(as $\mathbf M_{G_<}\in \mathsf X_{t'}(\mathbf M_G)$)}\\
&\leq (F_{\mathsf T_{\rm univ}}\circ F_{\Str}\circ F_{\mathsf X_{t'}})(k).
\end{align*}

Hence, ${\rm mw}_1(\mathsf T_{\rm univ}(\Dd))$  is bounded, which implies that $\mathsf T_{\rm univ}(\Dd)$ is not the class of all graphs. Then, it follows from \zcref{lem:tww_uniqueT} that~$\Dd$ and its reduct $\Cc$ have bounded twin-width.
\end{proof}

For example, as every tournament has an optimal loopless merge-model for every~$r$, we can directly deduce from \zcref{thm:mw_tww} that a class of tournaments has bounded twin-width if and only if it has bounded merge-width. However, this is also a direct consequence of the  known stronger property that a class of tournaments has bounded twin-width if and only if it is monadically dependent \cite{geniet2023first}.
\medskip

Note that the constant $r_0$ that can be derived from the proof of \zcref{thm:mw_tww} is very large, and this naturally raises the question of the optimal value of $r_0$.

\section{Conclusion}
In this paper, we prove that every graph with radius-$r$ merge-width $t$ has a merge-model with radius-$r$ merge-width at most $2t$. This complements the  proof given in \cite{lmcs_perm} that every graph with twin-width $t$ has a twin-model with twin-width at most $2t$. Moreover, as shown in \zcref{sec:more}, every graph with clique-width (resp. linear clique-width) $t$ has a twin-model with clique-width (resp. linear clique-width) at most $2t$.
These properties witness that some characteristic properties of unstable monadically dependent classes can be preserved when considering tree-ordered weakly sparse models, which justifies studying them for their own sake \cite{tows_arxiv}.

\bibliographystyle{amsplain}
\bibliography{ref}
\end{document}